\documentclass[aps,pre,onecolumn,superscriptaddress,amsmath,amssymb]{revtex4-2}

\usepackage{graphicx}
\usepackage{amsmath}
\usepackage{amssymb}
\usepackage{amsfonts}
\usepackage{amsthm}
\usepackage[T1]{fontenc}
\usepackage[]{xcolor}
\usepackage{diagbox}
\usepackage{mathrsfs}
\usepackage{booktabs}
\usepackage{multirow}
\usepackage{tikz}
\usepackage[hypertexnames=false]{hyperref}
\usepackage{enumitem}
\usepackage{float}
\newcommand \bl{\color{blue}}
\newcommand \rd{\color{red}}
\newcommand \gr{\color{green}}
\newcommand{\su}[1]{\underline{#1}}
\newcommand{\du}[1]{\underline{\underline{#1}}}
\newcommand{\tu}[1]{\underline{\underline{\underline{#1}}}}

\newtheorem{theorem}{Theorem}[section]
\newtheorem{lemma}[theorem]{Lemma}
\newtheorem{proposition}[theorem]{Proposition}
\newtheorem{remark}[theorem]{Remark}
\newtheorem{definition}[theorem]{Definition}
\newtheorem{example}[theorem]{Example}

\begin{document}

\title{Invariant Measures for Soliton Systems Generated by Mealy Automata}

\author{Takahiro Kanazawa}
\affiliation{Department of Physics, The University of Tokyo, 7-3-1 Hongo, Bunkyo-ku, Tokyo 113-0033, Japan}
\affiliation{Nonequilibrium Physics of Living Matter Laboratory, RIKEN Pioneering Research Institute, 2-1 Hirosawa, Wako 351-0198, Japan}
\affiliation{Department of Physics, The University of Chicago, 5720 South Ellis Avenue, Chicago, Illinois 60637, USA}
\author{Yutaro Nakabayashi}
\affiliation{Graduate School of Mathematical Sciences, The University of Tokyo, 3-8-1 Komaba, Meguro-ku, Tokyo 153-8914, Japan}

\date{\today}

\begin{abstract}
    We study invariant measures for soliton systems described by Mealy automata. Motivated by recently introduced soliton models associated with 2-letter, 3-state Mealy automata, we formulate the time evolution induced by Mealy automata on bi-infinite configuration spaces. We provide sufficient conditions for the invariance of Bernoulli product measures and derive a criterion for the invariance of two-sided space-homogeneous Markov distributions. We then apply these general results to three soliton models, which can be interpreted as variants of the box-ball system (BBS). For two of these models, BBS-S(2) and BBS-V(2), we prove that Bernoulli product measures are invariant. For the remaining model, BBS-C(2), we establish a more general result: the invariance of two-sided space-homogeneous Markov distributions, which include Bernoulli product measures as a special case. Furthermore, for all three models, we compute the phase shift associated with the interaction of two solitons, as well as the velocity of an isolated soliton. Although the latter has already been studied previously, both quantities constitute fundamental characteristics for understanding the generalized hydrodynamics of these systems. These results provide a foundation for the study of invariant measures, generalized Gibbs ensembles, and generalized hydrodynamic behavior in Mealy-automaton soliton systems.
\end{abstract}

\maketitle

\section{Introduction}

The study of invariant measures plays a central role in statistical physics, ergodic theory, and the theory of integrable systems~\cite{novikov1984theory, walters2000introduction, friedli2017statistical}. In recent years, the generalized Gibbs ensemble (GGE) has been actively studied in the context of statistical physics and integrable systems~\cite{spohn2020generalized, kuniba2020generalized}. Integrable systems are characterized by the existence of many, often infinitely many, conserved quantities, in contrast to non-integrable systems. The GGE extends the Gibbs (Boltzmann) distribution by assigning a weight to each conserved quantity. Thus, the GGE involves multiple parameters reflecting the number of conserved quantities. Correspondingly, invariant distributions of integrable systems are generally not uniquely determined; rather, one expects a family of invariant measures with degrees of freedom associated with these conserved quantities. In parallel with the development of the GGE, research on generalized hydrodynamics (GHD), a hydrodynamic equation for integrable systems based on the GGE, has also been advancing~\cite{doyon2019generalized, kuniba2020generalized}.

A natural class of systems to study within the framework of integrable systems and their associated invariant measures consists of bi-infinite binary sequences $\eta=\{\eta_n\}_{n\in \mathbb{Z}}$, which we will analyze throughout this paper. We interpret $\eta_n=1$ as the presence of a particle (or a ball) at site $n$, and $\eta_n=0$ as its absence. Each $n\in\mathbb{Z}$ is referred to as a site on a one-dimensional lattice. Hereafter, we use the terms ``particle'' and ``ball'' interchangeably.

As a pedagogical example to illustrate the dynamics of binary sequences, we consider the box-ball system (BBS). The BBS, introduced by Takahashi and Satsuma~\cite{19903514}, is a classic model of discrete integrable systems. We describe the dynamics of the BBS using the carrier formulation~\cite{DaisukeTakahashi_1997}. Consider a one-dimensional array of boxes (or sites) indexed by $\mathbb{Z}$. Balls (or particles) occupy only finitely many sites, and all other sites are empty. Since only a finite number of sites are occupied, we can choose $N\in \mathbb{Z}$ such that $\eta_n = 0$ for all $n<N$. We begin with an empty carrier at site $N$ and move it along the sites $N, N+1, N+2, \ldots$. When the carrier passes a site that contains a ball, it picks up the ball. When, in turn, the carrier reaches an empty site and is carrying at least one ball, it drops one ball into the site. This operation, performed sequentially across the relevant sites, produces a new configuration $\tilde{\eta}$ from the original configuration $\eta$. We regard $\tilde{\eta}$ as the configuration resulting from $\eta$ after one unit of time and refer to this map as the time evolution. In~\cite{croydon2023dynamics, ferrari2021soliton}, the definition of the BBS is naturally extended, under suitable conditions, to the case where infinitely many sites are occupied. This extension is natural from the viewpoint of invariant measures, since in a finite-particle system, all particles drift to the right, and hence no nontrivial invariant measure is expected.

The dynamics of the BBS exhibits solitons: stable clusters of balls and boxes that retain their shape under time evolution, as illustrated by the blue and green blocks in Fig.~\ref{fig:time_evolution_bbs}. These solitons appear as blocks of the form $1^k0^k$ or $0^k1^k$, meaning $k$ consecutive $1$'s followed by $k$ consecutive $0$'s, or vice versa. We refer to a soliton of size $k$ as a $k$-soliton. A free $k$-soliton has velocity $k$, which we call its bare velocity. This is the velocity of the soliton in the absence of interactions with other solitons. When a faster $k$-soliton overtakes a slower $l$-soliton with $k>l$, the solitons experience phase shifts, namely displacements relative to the non-interacting evolution: the $k$-soliton is shifted to the right by $2l$, while the $l$-soliton is shifted to the left by $2l$, relative to the non-interacting evolution. Solitons can be precisely identified using the Takahashi-Satsuma algorithm~\cite{19903514}. In the BBS, the numbers of solitons of each size are conserved quantities. This conservation property leads to a decomposition of any configuration into components corresponding to solitons of each size. In~\cite{ferrari2021soliton}, it is shown that one can construct a large family of invariant measures by prescribing distributions for these components under appropriate conditions. Regarding hydrodynamic behavior, in~\cite{Croydon2021}, the GHD equation for the BBS on the half-line is derived under suitable assumptions via an Euler-scaling limit.
The BBS serves as a prototypical example within the Mealy automaton~\cite{Mealy1955} framework, in which the output symbol (e.g., a letter) and the next state are determined by the current state and the input symbol (see Example~\ref{ex:TakahashiSatsuma}).

\begin{figure}[htbp]
    \centering
    \[
    \text{BBS}\ 
    \begin{array}{lll}
        \text{t=0:} & \text{configuration} & \texttt{0000{\bl\du{11110000}}0{\gr\su{1100}}000000000000000000} \\
        & \rd{\text{carrier}} & \rd{\texttt{00001234321001210000000000000000000}} \\
        \text{t=1:} & \text{configuration} & \texttt{00000000{\bl\du{1111000}{\gr\su{1100}}\du{0}}000000000000000} \\
        & \rd{\text{carrier}} & \rd{\texttt{00000000123432123210000000000000000}} \\
        \text{t=2:} & \text{configuration} & \texttt{000000000000{\bl\du{111}{\gr\su{0011}}\du{10000}}00000000000} \\
        & \rd{\text{carrier}} & \rd{\texttt{00000000000012321234321000000000000}} \\
        \text{t=3:} & \text{configuration} & \texttt{000000000000000{\gr\su{1100}}0{\bl\du{11110000}}0000000} \\
    \end{array}
    \]
    \caption{Time evolution of the box--ball system (BBS) with a carrier. The row labeled ``carrier'' records the number of balls carried after passing each site. The single-underlined (green) and double-underlined (blue) blocks indicate solitons of sizes $2$ and $4$, respectively. At $t=3$, the interaction between the solitons is complete and they are fully separated}
    \label{fig:time_evolution_bbs}
\end{figure}

In~\cite{maeno2025solitons}, three new models with soliton solutions described by Mealy automata were introduced. These models may be regarded as variants of the BBS in which the carrier rule is slightly modified. More precisely, they are described by 2-letter, 3-state Mealy automata that are bijective, transitive, particle-preserving, and locally interacting. While a large family of invariant measures for the BBS has been constructed based on the soliton decomposition~\cite{ferrari2021soliton}, the structure of invariant measures for soliton systems described by Mealy automata remains largely unexplored. The aim of this paper is to develop a general framework for studying invariant measures for such systems and to apply it to the three models from~\cite{maeno2025solitons} as concrete examples. To this end, we first formulate the dynamics induced by a Mealy automaton on bi-infinite configuration spaces and provide sufficient conditions under which Bernoulli product measures or two-sided Markov distributions are invariant. We then use these criteria to compute explicit classes of invariant Bernoulli and Markov measures for the three models. In addition, we analyze the corresponding soliton dynamics by computing soliton velocities and phase shifts for these models.

A natural direction for future work is to clarify the relation between the invariant measures obtained here and generalized Gibbs ensembles, as well as to investigate the generalized hydrodynamic behavior of these models. We expect that the three models will offer a convenient framework for testing ideas related to GGE and GHD, and may also provide insight into broader classes of integrable systems described by automata.

The rest of this paper is organized as follows. In Section~\ref{sec:automata}, we formulate the time evolution induced by a Mealy automaton on a suitable configuration space, and provide sufficient conditions under which Bernoulli product measures (Theorem~\ref{thm:Ber_inv_meas}) or two-sided Markov distributions (Theorem~\ref{thm:inv_meas}) are invariant under the automaton dynamics. In Section~\ref{sec:models}, we describe the three models BBS-C(2), BBS-S(2), and BBS-V(2) in detail and specify the relevant configuration spaces; we also verify that the dynamics preserve these spaces (Proposition~\ref{prop:preserveOmega}). In Section~\ref{sec:solitons}, we discuss solitons in these models, including their velocities and the phase shifts resulting from soliton interactions. In Section~\ref{sec:measures}, we apply the general results from Section~\ref{sec:automata} to the three models and prove the following statements about invariant measures: Bernoulli product measures are invariant for BBS-S(2) and BBS-V(2) (Theorem~\ref{inv_meas_S2V2}), while a two-sided space-homogeneous Markov distribution (which includes Bernoulli product measures as a special case) is invariant for BBS-C(2) (Theorem~\ref{thm:inv_meas_C2}). Finally, the appendices contain the proof of a technical proposition concerning the Markov distribution (Appendix~\ref{app:Mar_dist_supp}) and the proof of Remark~\ref{rem:Omega_ext}, which explains the
choice of the configuration space (Appendix~\ref{app:Omega_ext}).

\section{Time Evolution and Invariant Measures in Mealy Automata}
\label{sec:automata}

In this section, we introduce the formalism of Mealy automata, which governs the time evolution of configurations. We provide precise definitions and describe the operational rules by which Mealy automata generate output sequences. Furthermore, we establish general conditions under which Bernoulli product measures or two-sided Markov distributions are invariant with respect to the dynamics. These results serve as a framework for the analysis of the specific models in the subsequent sections.

\subsection{Mealy Automata and Their Properties}

Let $Q$ be a non-empty countable set of states and $S$ be a non-empty set of letters, called an alphabet. In this paper, we focus on the two-letter case $S = \{0,1\}$, meaning that a site is either vacant (represented by $0$) or occupied (represented by $1$). Our definition of a Mealy automaton follows~\cite{maeno2025solitons}. A Mealy automaton is characterized by two functions: a transition function $\varphi$ and an output function $\psi$. The transition function $\varphi$ maps a state and an input letter to the next state:
\begin{equation}
    \varphi\colon Q\times S \to Q.
\end{equation}
The output function $\psi$ maps a state and an input letter to an output letter:
\begin{equation}
    \psi\colon Q \times S \to S.
\end{equation}
The automaton is formally defined by the quadruple $(Q, S, \varphi, \psi)$. The actions of $\varphi$ and $\psi$ are illustrated in Fig.~\ref{fig:action}.

\begin{figure}[htbp]
\centering
\begin{tikzpicture}{auto}
    \node (s) at (0,0) {$s$};
    \node (q) at (-1,1) {$q$};
    \node (stilde) at (0,2) {$\tilde{s}$};
    \node (tildeq) at (1,1) {$\tilde{q}$};

    \draw[->] (s) to node[xshift = -5pt, yshift = 10pt] {$\scriptstyle \psi$} (stilde);
    \draw[->] (q) to node[xshift = 10pt, yshift = 5pt] {$\scriptstyle \varphi$} (tildeq);
\end{tikzpicture}
\caption{Action of $(\varphi, \psi)$}
\label{fig:action}
\end{figure}
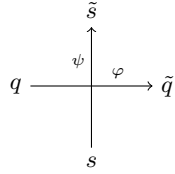

\noindent
Here, $q, \tilde{q}\in Q$, $s,\tilde{s}\in S$, and 
\begin{equation}
    \left\{
    \begin{aligned}
        \varphi(q,s) &= \tilde{q} \\
        \psi(q, s)   &= \tilde{s}.
    \end{aligned}
    \right.
\end{equation}
The action of the automaton on a sequence of input letters is defined recursively. That is, given an initial state $Q_0\in Q$ and an input sequence of letters $s_1,s_2,\ldots, s_n \in S$, the sequence of states $Q_1,Q_2,\ldots, Q_n$ and the output sequence of letters $\tilde{s}_1, \tilde{s}_2,\ldots, \tilde{s}_n$ are generated by the relations
\begin{equation}
    \left\{
    \begin{aligned}
    Q_{k} &= \varphi(Q_{k-1},s_k) \\
    \tilde{s}_k &= \psi(Q_{k-1}, s_k)
    \end{aligned}
    \right.
\end{equation}
for $k= 1,2,\ldots, n$. The action on an input sequence is illustrated in Fig.~\ref{fig:action_on_seq}.

\begin{figure}[htbp]
\centering
    \begin{tikzpicture}
        \node (s1) at (0,0) {$s_1$};
        \node (s2) at (2,0) {$s_2$};
        \node (Q0) at (-1,1) {$Q_0$};
        \node (Q1) at (1,1) {$Q_1$};
        \node (stilde1) at (0,2) {$\tilde{s}_1$};
        \node (stilde2) at (2,2) {$\tilde{s}_2$};
        \node (Q2) at (3,1) {$Q_2$};
        \node (Qk-1) at (4.3,1) {$Q_{k-1}$};
        \node (sk) at (5.3,0) {$s_k$};
        \node (Qk) at (6.3,1) {$Q_k$};
        \node (stildek) at (5.3,2) {$\tilde{s}_k$};
        \node (Qn-1) at (7.6,1) {$Q_{n-1}$};
        \node (sn) at (8.6,0) {$s_n$};
        \node (Qn) at (9.6,1) {$Q_n$};
        \node (stilden) at (8.6,2) {$\tilde{s}_n$};

        \draw[->] (s1) to node[xshift = -5pt, yshift = 10pt] {$\scriptstyle \psi$} (stilde1);
        \draw[->] (s2) to node[xshift = -5pt, yshift = 10pt] {$\scriptstyle \psi$} (stilde2);
        \draw[->] (sk) to node[xshift = -5pt, yshift = 10pt] {$\scriptstyle \psi$} (stildek);
        \draw[->] (Q0) to node[xshift = 10pt, yshift = 5pt] {$\scriptstyle \varphi$} (Q1);
        \draw[->] (Q1) to node[xshift = 10pt, yshift = 5pt] {$\scriptstyle \varphi$} (Q2);
        \draw[->] (Qk-1) to node[xshift = 10pt, yshift = 5pt] {$\scriptstyle 
        \varphi$} (Qk);
        \draw[->] (Qn-1) to node[xshift = 10pt, yshift = 5pt] {$\scriptstyle \varphi$} (Qn);
        \draw[->] (sn) to node[xshift = -5pt, yshift = 10pt] {$\scriptstyle \psi$} (stilden);
        \draw[loosely dotted] (Q2) -- (Qk-1);
        \draw[loosely dotted] (Qk) -- (Qn-1);
    \end{tikzpicture}
\centering
\caption{Action of $(\varphi, \psi)$ on a sequence}
\label{fig:action_on_seq}
\end{figure}
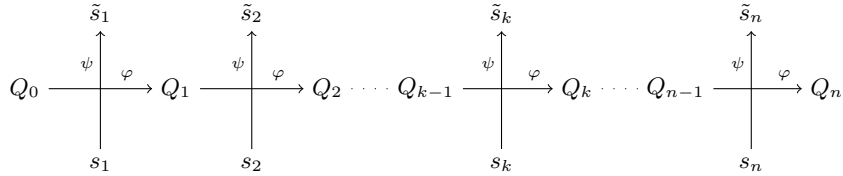

Here we introduce several concepts concerning automata. We follow the definitions in~\cite{maeno2025solitons}. The particle-preserving property corresponds to the conservation of the number of particles in our framework, while bijectivity implies that the dynamics admit an inverse operation.

\begin{definition}
    An automaton $\mathscr{A}=(Q, S, \varphi, \psi)$ is said to be particle-preserving if there exists a weight function $w\colon Q\cup S\to \mathbb{R}$ such that $w$ is not constant on $S$, and the following conservation law holds for all $(q,s)\in Q \times S$:
    \begin{equation}
        w(q)+w(s) = w(\varphi(q,s)) + w(\psi(q,s)).
        \label{eq:conserv}
    \end{equation}
    \label{def:weight}
\end{definition}

\begin{remark}\label{rem:particle_preserving}
    When $S=\{0,1\}$, the weight function $w$ may, without loss of generality, be chosen so that $w(s)=s$ on $S$. In fact, by the definition of particle preservation, $w$ is not constant on $S$, that is, $w(0)\neq w(1)$. Any such function can be written as an affine function $w(s) = as + b$ for some constants $a,b\in \mathbb{R}$ with $a\neq 0$. Substituting this into the conservation law~\eqref{eq:conserv} gives
    \begin{equation}
        w(q) + as + b=w(\varphi(q,s))+a\psi(q,s) + b.
    \end{equation}
    By redefining the weight function as
    \begin{equation}
    \left\{
        \begin{aligned}
            w'(q) &=\frac{w(q)}{a}\\
            w'(s) &= s,
        \end{aligned}
    \right.
    \end{equation}
    we obtain
    \begin{equation}
        w'(q)+w'(s)=w'(\varphi(q,s)) + w'(\psi(q,s)).
    \end{equation}
    This is precisely the conservation law for the case where the weight function is chosen as $w(s) = s$ on $\{0,1\}$. Moreover, there is no loss of generality in assuming $w(\tilde{q})=0$ for an arbitrarily chosen $\tilde{q}$ by redefining the weight function as
    \begin{equation}
    \left\{
    \begin{aligned}
        w'(q) &= w(q) - w(\tilde{q}) \\
        w'(s) &= w(s).
    \end{aligned}
    \right.
    \end{equation}
\end{remark}

\begin{definition}
    An automaton $\mathscr{A} = (Q,S,\varphi, \psi)$ is said to be bijective if the mapping
    \begin{equation}
        (q,s)\longmapsto (\varphi(q,s), \psi(q,s))
    \end{equation}
    is a bijection on $Q\times S$. When $\mathscr{A}$ is bijective, we denote its inverse by $(q,s) \mapsto (\tilde{\varphi}(q,s),\tilde{\psi}(q,s))$.
\end{definition}

\subsection{Configuration Space and Well-Defined Action}\label{sec:config}

Let $\mathscr{A} = (Q,\{0,1\}, \varphi, \psi)$ be an automaton. From this point onward, we write a finite word $(s_1,s_2,\ldots,s_n)\in \{0,1\}^n$ as $s_1s_2\ldots s_n$. We also use the shorthand $0^n$ (resp. $1^n$) for a block of $n$ consecutive zeros (resp. ones). For example, $111000011=1^30^41^2$. For a finite word $\mathbf{s} = s_1s_2\cdots s_n \in \{0,1\}^n$ and an initial state $q\in Q$, we define $\varphi(q;\mathbf{s})\in Q$ and $\psi(q;\mathbf{s})\in \{0,1\}^n$ as follows: $\varphi(q;\mathbf{s})$ is the final state obtained by scanning $\mathbf{s}$ from left to right starting from the state $q$, while $\psi(q;\mathbf{s})$ denotes the corresponding output sequence. For example, in Fig.~\ref{fig:action_on_seq}, we have
\begin{equation}
    \left\{
    \begin{aligned}
    \varphi(Q_0; s_1s_2\cdots s_n) &= Q_n \\
    \psi(Q_0;s_1s_2\cdots s_n) &= \tilde{s}_1\tilde{s}_2\cdots \tilde{s}_n.
    \end{aligned}
    \right.
\end{equation}
When $\mathscr{A}$ is bijective, for $\tilde{\mathbf{s}}\in \{0,1\}^n$ and $q\in Q$, we denote by $\tilde{\varphi}(\tilde{\mathbf{s}};q)\in Q$ and $\tilde{\psi}(\tilde{\mathbf{s}};q)\in \{0,1\}^n$ the final state and output sequence obtained by scanning $\tilde{\mathbf{s}}$ from right to left starting from state $q$ with the inverse functions $(\tilde{\varphi}, \tilde{\psi})$. Figure~\ref{fig:inverse} illustrates the relations $\tilde{\varphi}(\tilde{s}_1 \tilde{s}_2\cdots\tilde{s}_n;Q_n) = Q_0$ and $\tilde{\psi}(\tilde{s}_1\tilde{s}_2\cdots\tilde{s}_n;Q_n) = s_1s_2\cdots s_n$. In particular, for any $q\in Q$ and $\mathbf{s}\in\{0,1\}^n$,
\begin{equation}
    (\tilde{\varphi}(\psi(q; \mathbf{s}); \varphi(q; \mathbf{s})), \tilde{\psi}(\psi(q; \mathbf{s}); \varphi(q; \mathbf{s})))=(q, \mathbf{s}).
\end{equation}

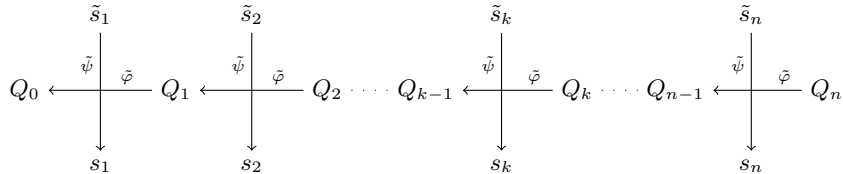
\begin{figure}[htbp]
\centering
    \begin{tikzpicture}
        \node (s1) at (0,0) {$s_1$};
        \node (s2) at (2,0) {$s_2$};
        \node (Q0) at (-1,1) {$Q_0$};
        \node (Q1) at (1,1) {$Q_1$};
        \node (stilde1) at (0,2) {$\tilde{s}_1$};
        \node (stilde2) at (2,2) {$\tilde{s}_2$};
        \node (Q2) at (3,1) {$Q_2$};
        \node (Qk-1) at (4.3,1) {$Q_{k-1}$};
        \node (sk) at (5.3,0) {$s_k$};
        \node (Qk) at (6.3,1) {$Q_k$};
        \node (stildek) at (5.3,2) {$\tilde{s}_k$};
        \node (Qn-1) at (7.6,1) {$Q_{n-1}$};
        \node (sn) at (8.6,0) {$s_n$};
        \node (Qn) at (9.6,1) {$Q_n$};
        \node (stilden) at (8.6,2) {$\tilde{s}_n$};
        
        \draw[<-] (s1) to node[xshift = -5pt, yshift = 10pt] {$\scriptstyle \tilde{\psi}$} (stilde1);
        \draw[<-] (s2) to node[xshift = -5pt, yshift = 10pt] {$\scriptstyle \tilde{\psi}$} (stilde2);
        \draw[<-] (sk) to node[xshift = -5pt, yshift = 10pt] {$\scriptstyle \tilde{\psi}$} (stildek);
        \draw[<-] (Q0) to node[xshift = 10pt, yshift = 5pt] {$\scriptstyle \tilde{\varphi}$} (Q1);
        \draw[<-] (Q1) to node[xshift = 10pt, yshift = 5pt] {$\scriptstyle \tilde{\varphi}$} (Q2);
        \draw[<-] (Qk-1) to node[xshift = 10pt, yshift = 5pt] {$\scriptstyle \tilde{\varphi}$} (Qk);
        \draw[<-] (Qn-1) to node[xshift = 10pt, yshift = 5pt] {$\scriptstyle \tilde{\varphi}$} (Qn);
        \draw[<-] (sn) to node[xshift = -5pt, yshift = 10pt] {$\scriptstyle \tilde{\psi}$} (stilden);
        
        \draw[loosely dotted] (Q2) -- (Qk-1);
        \draw[loosely dotted] (Qk) -- (Qn-1);
    \end{tikzpicture}
\caption{Operation of $\tilde{\varphi}$ and $\tilde{\psi}$}
\label{fig:inverse}
\end{figure}

\begin{definition}
    We define a subset $\Omega_{\mathscr{A}}\subset \{0,1\}^\mathbb{Z}$ by
    \begin{equation}
    \begin{aligned}
    \Omega_{\mathscr{A}}
    =\left\{
    \eta \in \{0,1\}^{\mathbb{Z}}
    \left| \
    \begin{aligned}
    &\forall n \in \mathbb{Z}, \ \forall q \in Q, \
    Q_n := \lim_{k \to \infty}
    \varphi\bigl(q;\eta_{n-k}\eta_{n-k+1}\cdots\eta_n\bigr)\\
    &\text{exists} \
    \text{and is independent of } q.
    \end{aligned}
    \right.
    \right\}.
    \end{aligned}
    \end{equation}
    Here, the existence of the limit means that there exists $k_0 \in \mathbb{Z}_{\ge 0}$ such that
    \begin{equation}
    \varphi\bigl(q;\eta_{n-k}\eta_{n-k+1}\cdots\eta_n\bigr)
    = Q_n
    \end{equation}
    holds for all $k \ge k_0$. We define a map
    $T_{\mathscr{A}} \colon \Omega_{\mathscr{A}} \to \{0,1\}^{\mathbb{Z}}$
    by
    \begin{equation}
    \bigl(T_{\mathscr{A}}(\eta)\bigr)_n
    = \psi(Q_{n-1}, \eta_n).
    \end{equation}
    \end{definition}
    The set $\Omega_{\mathscr{A}}$ consists of configurations for which the action of the automaton $\mathscr{A}$ is well-defined. We interpret $T_\mathscr{A}$ as the one-step time evolution: $\eta$ represents the initial configuration at $t=0$. If $T_\mathscr{A}(\Omega_\mathscr{A})\subset \Omega_\mathscr{A}$, then the iterates $T_\mathscr{A}^t(\eta)$ are well-defined for $t\in \mathbb{Z}_{\geq 2}$, and $T_\mathscr{A}^t(\eta)$ represents the configuration at time $t$. If $\mathscr{A}$ is bijective, then the inverse operation gives a one-step backward evolution on the set of configurations for which the action of the inverse automaton is well-defined. 

\begin{example}\label{ex:TakahashiSatsuma}
The box-ball system~\cite{19903514} can be represented as an automaton with $Q = \mathbb{Z}_{\geq 0}$, $S=\{0,1\}$, and
\begin{equation}
\left\{
    \begin{aligned}
        \varphi(q,s) &=
        \begin{cases}
            q+1 & \text{if}\ s = 1\\
            q-1 & \text{if}\ q > 0, s = 0\\
            0   & \text{if}\ q = 0, s=0
        \end{cases}\\
        \psi(q,s) &=
        \begin{cases}
            0 & \text{if}\ s = 1\\
            0 & \text{if}\ q=0, s=0 \\
            1 & \text{if}\ q>0, s=0.
        \end{cases}
    \end{aligned}
\right.
\end{equation}
That is, when scanning from left to right, the carrier picks up a ball at each occupied site $(\text{corresponding to} \ 1\mapsto 0 \ \text{in} \ S)$. If the carrier is nonempty, it deposits a ball at the next vacant site $(0\mapsto 1)$, otherwise the site remains vacant $(0\mapsto 0)$. For configurations with finitely many balls, the box-ball dynamics are obtained by placing an empty carrier immediately to the left of the leftmost ball and scanning to the right according to the carrier rule. 

In~\cite{croydon2023dynamics}, the dynamics are extended to configurations with infinitely many balls. Define
\begin{equation}
    \Omega_{\textup{sub-critical}}^- = \left\{\eta\in\{0,1\}^\mathbb{Z} \left|\ \lim_{n\to \infty} \left(\sum_{m=-n+1}^0\eta_m - \frac{n}{2}\right)\to -\infty\right.\right\}.
\end{equation}
Then $\Omega_{\textup{sub-critical}}^- \subset \Omega_{\mathscr{A}}$. Left sub-criticality of the path ensures that the left boundary limit of the carrier exists and is independent of the initial state, so the time evolution determined by $\mathscr{A}$ is well-defined.
\end{example}

\subsection{Bernoulli Distributions, Markov Distributions, and Their Invariance under Automaton Dynamics}

In this section, we investigate the invariance of Bernoulli and Markov distributions on $\{0,1\}^\mathbb{Z}$ under the time evolution $T_\mathscr{A}$ induced by the automaton $\mathscr{A} = (Q, \{0,1\}, \varphi, \psi)$. We consider a Markov distribution $\nu_M$ on $\{0,1\}^\mathbb{Z}$ characterized by a transition probability matrix $P$. The transition probability matrix $P$ is parameterized by $p, r \in (0, 1)$ as follows:
\begin{equation}
P =
\begin{pmatrix}
p & 1-p\\
1-r & r
\end{pmatrix}.
\end{equation}
Here, $p$ denotes the transition probability from $0$ to $0$, while $r$ is the transition probability from $1$ to $1$. When $r = 1-p$, the Markov distribution $\nu_M$ reduces to the product Bernoulli distribution $\nu_p = \operatorname{Be}(p)^{\otimes\mathbb{Z}}$. Since it suffices to consider the dynamics on a subset on which $\nu_M$ is concentrated, we restrict our attention to the following tractable configuration space $\Omega$ defined as
\begin{align}
\label{def:Omega}
\Omega 
&= \Omega_+ \cap \Omega_-
\end{align}
with
\begin{align}
\Omega_+ 
&= \bigcap_{N = 1}^{\infty} \bigcap_{M = 1}^{\infty}
\bigcup_{m = M}^{\infty} A_{m, N}
= \bigcap_{N = 1}^{\infty}
\limsup_{m \to \infty} A_{m, N}\\
\Omega_- 
&= \bigcap_{N = 1}^{\infty} \bigcap_{M = 1}^{\infty}
\bigcup_{m = M}^{\infty} B_{m, N}
= \bigcap_{N = 1}^{\infty}
\limsup_{m \to \infty} B_{m, N} ,
\end{align}
where
\begin{align}
A_{m, N} 
&= \{\eta \in \{0,1\}^{\mathbb{Z}} \mid
\eta_{m} = \ldots = \eta_{m + N - 1} = 0\}\\
B_{m, N} 
&= \{\eta \in \{0,1\}^{\mathbb{Z}} \mid
\eta_{-m} = \ldots = \eta_{-m - N + 1} = 0\}.
\end{align}
Here, $A_{m, N}$ denotes the set of configurations containing a block of $N$ consecutive empty sites ($0$'s) starting at position $m$ to the right. Similarly, $B_{m, N}$ denotes the set of configurations containing a block of $N$ consecutive empty sites ($0$'s) starting at position $-m$ to the left. Consequently, $\Omega_+$ is the set of configurations where arbitrarily long blocks of zeros appear infinitely often in the sequence $(\eta_n)_{n \in \mathbb{Z}_{\geq 0}}$, while $\Omega_-$ consists of configurations for which arbitrarily long blocks of zeros appear infinitely often in the sequence $(\eta_n)_{n \in \mathbb{Z}_{\leq 0}}$. Since $\Omega$ is their intersection, it is the set of configurations where arbitrarily long blocks of consecutive empty sites appear infinitely often in both the positive and negative directions.
\begin{proposition}\label{prop:Mar_dist_supp}
The Markov distribution $\nu_M$ is supported on $\Omega$. That is,
\begin{equation}
\nu_M(\Omega) = 1.
\end{equation}
\end{proposition}
\begin{proof}
See Appendix~\ref{app:Mar_dist_supp}.
\end{proof}

We first give a sufficient condition for a probability measure on $\{0,1\}^\mathbb{Z}$ to be a Markov distribution. This criterion will later be used to prove invariance of the Markov distribution under the automaton dynamics.
\begin{lemma}\label{lem:000}
    Let $\mu$ be a shift-invariant probability measure on $\{0,1\}^{\mathbb{Z}}$ with $\mu(\Omega)=1$.
    Assume that there exists an integer $m\geq 1$ such that, for each $j\in \{0,1\}$, for all $l \geq m+2$, and for any sequence $\mathbf{r} \in \{0,1\}^{l-m-1}$, the following equation holds:
    \begin{align}
        \mu(\eta_{l+1}=j \mid \eta_1\cdots\eta_l = 0^m\mathbf{r}j)=  \mu(\eta_{l+1}=j\mid \eta_l=j).
        \label{eq:mu_assum}
    \end{align}
    Then $\mu$ is a Markov distribution. That is, for any integer $n\geq 2$, for each $j\in\{0,1\}$, and for any sequence $\mathbf{s}\in \{0,1\}^{n-1}$, the equation
    \begin{equation}
        \mu(\eta_{n+1}=j\mid \eta_1\cdots\eta_n=\mathbf{s}j)
        =\mu(\eta_{n+1}=j\mid\eta_n=j)
        \label{eq:mu_goal}
    \end{equation}
    also holds.
\end{lemma}
\begin{proof}
    By assumption, we may choose an integer $m\geq 1$ such that \eqref{eq:mu_assum} holds for each $j\in\{0,1\}$, for all $l\geq m+2$, and for any sequence $\mathbf{r}\in\{0,1\}^{l-m-1}$. Fix such an $m$ throughout the proof. Let $\mathbf{s}'= \mathbf{s}j=s_1\cdots s_{n-1} j$. We compute the conditional probability appearing on the left-hand side of \eqref{eq:mu_goal},
    \begin{equation}
        \mu(\eta_{n+1}=j\mid \eta_1 \cdots \eta_n=\mathbf{s}').
    \end{equation}
    
    First, for each integer $k \leq 0$, we define $S_k(\mathbf{s}')$ as
    \begin{equation}
        S_k(\mathbf{s}') = \{\tilde{\mathbf{s}} = s_{k-m+1}\cdots s_n\in \{0,1\}^{n-k+m} \mid s_1 \cdots s_n=\mathbf{s}', \max\{l\leq 0 \mid s_{l-m+1}\cdots s_{l}=0^m\} =k\}.
    \end{equation}
    For any $\eta \in \Omega$ satisfying $\eta_1\cdots \eta_n=\mathbf{s}'$, there exists a unique $k\leq 0$ such that $\eta_{k-m+1} \cdots \eta_n\in S_k(\mathbf{s}')$. Therefore,
    \begin{align}
        \mu(\eta_{n+1}=j, \eta_1 \cdots \eta_n= \mathbf{s}')
        &= \sum_{k\leq 0}\sum_{\tilde{\mathbf{s}}\in S_k(\mathbf{s}')} \mu(\eta_{n+1}=j, \eta_{k-m+1} \cdots \eta_n = \tilde{\mathbf{s}}) \\
        &=\mu(\eta_1 \cdots \eta_n=\mathbf{s}')\sum_{k\leq 0}\sum_{\tilde{\mathbf{s}}\in S_k(\mathbf{s}')} \mu(\eta_{n+1}=j\mid \eta_{k-m+1} \cdots \eta_n = \tilde{\mathbf{s}}) \mu(\eta_{k-m+1} \cdots \eta_n = \tilde{\mathbf{s}}\mid \eta_1 \cdots \eta_n=\mathbf{s}').
    \end{align}
    Then, by assumption \eqref{eq:mu_assum},
    \begin{align}
        \mu(\eta_{n+1}=j \mid \eta_1 \cdots \eta_n = \mathbf{s}')
        &=\frac{\mu(\eta_{n+1}=j, \eta_1 \cdots \eta_n= \mathbf{s}')}{\mu(\eta_1 \cdots \eta_n=\mathbf{s}')} \\
        &=\sum_{k\leq 0}\sum_{\tilde{\mathbf{s}}\in S_k(\mathbf{s}')} \mu(\eta_{n+1}=j\mid \eta_{k-m+1} \cdots \eta_n = \tilde{\mathbf{s}}) \mu(\eta_{k-m+1} \cdots \eta_n = \tilde{\mathbf{s}}\mid \eta_1 \cdots \eta_n=\mathbf{s}') \\
        &=\sum_{k\leq 0}\sum_{\tilde{\mathbf{s}}\in S_k(\mathbf{s}')} \mu(\eta_{n+1}=j\mid \eta_n = j) \mu(\eta_{k-m+1}\cdots \eta_n=\tilde{\mathbf{s}}\mid\eta_1 \cdots \eta_n=\mathbf{s}') \\
        &=\mu(\eta_{n+1}=j \mid \eta_n=j).
    \end{align}
    Here we used
    \begin{equation}
        \sum_{k\leq 0}\sum_{\tilde{\mathbf{s}}\in S_k(\mathbf{s}')} \mu(\eta_{k-m+1}\cdots \eta_n=\tilde{\mathbf{s}}\mid\eta_1 \cdots \eta_n=\mathbf{s}') = 1.
    \end{equation}
    This completes the proof.
\end{proof}

\begin{proposition}
    Assume $\nu_M(\Omega_\mathscr{A}) = 1$. Then, under $\nu_M$, the pair $X_n = (Q_{n-1},\eta_{n})$ is a $Q\times \{0,1\}$-valued Markov process. Moreover, the joint distribution of $(Q_{n-1}, \eta_n)$ is stationary for this Markov process.
\end{proposition}
\begin{proof}
    Notice that the following holds for any $n\in\mathbb{Z}$:
    \begin{equation}
    \left\{
    \begin{aligned}
        \nu_M(\eta_{n+1} = 0\mid \mathcal{F}_n) &= p\chi_{\{\eta_n = 0\}} + (1-p)\chi_{\{\eta_n = 1\}}\\
        \nu_M(\eta_{n+1} = 1\mid \mathcal{F}_n) &= (1-r)\chi_{\{\eta_n=0\}} +r\chi_{\{\eta_n = 1\}}.  
    \end{aligned}
    \right.
    \end{equation}
    Here, $\chi$ is an indicator function and $\mathcal{F}_n = \sigma(\{\eta_k\}_{k\leq n})$ is the $\sigma$-algebra generated by $\{\eta_k\}_{k\leq n}$. For any $(q_0,\ldots,q_{n})\in Q^{n+1}$ and $s_1\cdots s_{n+1}\in\{0,1\}^{n+1}$, we show that
    \begin{equation}
        \nu_M(X_{n+1} = (q_{n},s_{n+1})\mid X_{i+1} = (q_i,s_{i+1}), i= 0,\ldots,n-1) = \nu_M(X_{n+1} = (q_{n},s_{n+1})\mid X_n=(q_{n-1},s_{n})).
        \label{eq:Markov_proof}
    \end{equation}
    To prove \eqref{eq:Markov_proof}, we distinguish two cases. First, suppose that $q_{n} \neq \varphi(q_{n-1},s_{n})$. Then both sides of \eqref{eq:Markov_proof} are zero, so the equality holds. Next, assume that $q_n=\varphi(q_{n-1},s_n)$. Since $\{\eta_n\}_{n\in \mathbb{Z}}$ is a Markov chain and $Q_{n-1}$ is measurable with respect to $\sigma(\{\eta_k\}_{k\leq n})$,
    \begin{align}
    \nu_M(\eta_{n+1}=s_{n+1}\mid \eta_n=s_n)
    &=\nu_M(\eta_{n+1}=s_{n+1}\mid \eta_n=s_n,\; Q_{n-1}=q_{n-1}) \\
    &=\nu_M\bigl(X_{n+1}=(q_n,s_{n+1}) \,\big|\, X_n=(q_{n-1},s_n)\bigr).\label{eq:eta_and_X}
    \end{align}
    Since
    \begin{equation}
        \eta_{n+1}=s_{n+1}\text{ and } X_n=(q_{n-1},s_n)
    \end{equation}
    implies
    \begin{equation}
        X_{n+1}=(q_n,s_{n+1})\text{ and } 
        X_n = (q_{n-1},s_n),
    \end{equation}
    we obtain
    \begin{align}
        \nu_M(X_{i+1} = (q_i,s_{i+1}), i= 0,\ldots,n)
        &= E_{\nu_M}[\nu_M(\eta_{n+1}=s_{n+1}\mid \mathcal{F}_n) ; \{Q_n=q_n\}\cap\{X_{i+1}=(q_i,s_{i+1}) \ \text{for }i= 0,\ldots, n-1\}]\\
        &=\nu_M(\eta_{n+1}=s_{n+1}\mid \eta_n= s_n) 
        \nu_M(X_{i+1}=(q_i,s_{i+1}) \ \text{for }i= 0,\ldots, n-1).
        \label{eq:nu_decomposition}
    \end{align}
    From \eqref{eq:eta_and_X} and \eqref{eq:nu_decomposition}, we obtain \eqref{eq:Markov_proof}, and thus $\{X_n\}_{n\in\mathbb{Z}}$ is a Markov process. By the shift invariance of $\nu_M$, we have
    \begin{equation}
    X_n\stackrel{d}{=}X_{n+1},
    \end{equation}
    and hence the process is stationary.
\end{proof}

In this paper, we investigate invariant measures of $T_\mathscr{A}$ for automata $\mathscr{A}$ that satisfy the following conditions. The three models introduced in Section~\ref{sec:models} all fall within this framework.
\begin{enumerate}[label = (\textbf{A\arabic*})]
    \item $\mathscr{A}$ is bijective.\label{enum:bijective}
    \item There exists a unique $q^*\in Q$ such that $(\varphi(q^*,0), \psi(q^*, 0)) = (q^*, 0)$.\label{enum:unique_final_state}
    \item There exists $N\in \mathbb{Z}_{\geq 1}$ such that, for any $q\in Q$, $\varphi(q;0^N)=q^*$ with $q^*$ defined in Assumption~\ref{enum:unique_final_state}. If $\mathscr{A}$ is bijective, we additionally assume that, for every $q\in Q$, $\tilde{\varphi}(0^N;q) = q^*$.\label{enum:reset_word}
\end{enumerate}

\begin{proposition}\label{prop:TOmegainOmega}
    If an automaton $\mathscr{A}= (Q, \{0,1\}, \varphi, \psi)$ satisfies~\textup{\ref{enum:unique_final_state}} and \textup{\ref{enum:reset_word}}, $\Omega\subset \Omega_{\mathscr{A}}$ and $T_{\mathscr{A}}(\Omega)\subset\Omega$ hold.
\end{proposition}
\begin{proof}
    By~\ref{enum:reset_word}, there exists $N\in\mathbb{Z}_{\geq 1}$ such that, for any $q\in Q$, $\varphi(q; 0^N)=q^*$. If $\eta\in \Omega$, then $\eta$ contains infinitely many occurrences of $0^N$. Hence, $\Omega\subset\Omega_{\mathscr{A}}$.
    
    Next, to prove $T_{\mathscr{A}}(\Omega)\subset\Omega$, we first show that $T_\mathscr{A}(\Omega)\subset \Omega_+$. Let $\eta\in \Omega$ and $\tilde{\eta}=T_\mathscr{A}(\eta)$. It suffices to show that for any $L, M\in \mathbb {Z}_{\geq 1}$, there exists an integer $m\geq M$ such that
    \begin{equation}
        \tilde{\eta}_m\tilde{\eta}_{m+1}\cdots\tilde{\eta}_{m+L-1} = 0^L.
    \end{equation}
    Since $\eta\in \Omega$, there exists an integer $m_0\geq M$ such that
    \begin{equation}
        \eta_{m_0}\eta_{m_0+1}\cdots\eta_{m_0+N+L-1} = 0^{N+L}.
    \end{equation}
    For any $q\in Q$, \ref{enum:reset_word} implies that $\varphi(q;\eta_{m_0}\eta_{m_0+1}\cdots\eta_{m_0+N-1})=q^*$ holds, and hence, by \ref{enum:unique_final_state}, we obtain
    \begin{equation}
        \psi(q;\eta_{m_0}\eta_{m_0+1}\cdots\eta_{m_0+N+L-1}) = \tilde{\mathbf{s}}\psi(q^*;\eta_{m_0+N}\eta_{m_0+N+1}\cdots\eta_{m_0+N+L-1})=\tilde{\mathbf{s}}0^L.
    \end{equation}
    Here, $\tilde{\mathbf{s}}=\psi(q;\eta_{m_0}\eta_{m_0+1}\cdots\eta_{m_0+N-1})$ is an element of $\{0,1\}^N$. Therefore,
    \begin{equation}
        \tilde{\eta}_{m_0+N}\tilde{\eta}_{m_0+N+1}\cdots\tilde{\eta}_{m_0+N+L-1}=0^L.
    \end{equation}

    For $T_{\mathscr{A}}(\Omega)\subset\Omega_-$, it suffices to show that for any $L, M\in\mathbb{Z}_{\geq 1}$, there exists an integer $m\leq -M$ such that
    \begin{equation}
        \tilde{\eta}_{m-L+1}\tilde{\eta}_{m-L+2}\cdots\tilde{\eta}_m=0^L.
    \end{equation}
    Since $\eta\in \Omega$, there exists an integer $m_0\leq -M$ such that
    \begin{equation}
        \eta_{m_0 - N-L+1}\eta_{m_0-N-L+2}\cdots\eta_{m_0}=0^{N+L}.
    \end{equation}
    By the same argument as in the case of $\Omega_+$, we obtain
    \begin{equation}
        \tilde{\eta}_{m_0-L+1}\tilde{\eta}_{m_0-L+2}\cdots\tilde{\eta}_{m_0}=0^L.
    \end{equation}
    Therefore, from $T_{\mathscr{A}}(\Omega)\subset\Omega_+$ and $T_{\mathscr{A}}(\Omega)\subset\Omega_-$, we find $T_{\mathscr{A}}(\Omega)\subset\Omega$.
\end{proof}

\begin{remark}
    If an automaton $\mathscr{A}= (Q, \{0,1\}, \varphi, \psi)$ satisfies \textup{\ref{enum:unique_final_state}} and \textup{\ref{enum:reset_word}}, under $\nu_M$, the map $T_\mathscr{A}$ may be iterated indefinitely, that is, $T^t_\mathscr{A}(\eta)$ is well-defined for every $t\geq 0$ for $\nu_M$-a.e. $\eta$. The inclusion $\Omega\subset \Omega_{\mathscr{A}}$ does not hold for the \textup{BBS} automaton (Example~\ref{ex:TakahashiSatsuma}), since \textup{\ref{enum:reset_word}} is not satisfied. Indeed, let $\mathscr{A}_{\textup{BBS}}$ denote the Mealy automaton corresponding to the \textup{BBS}, and define a configuration $\eta=\{\eta_n\}_{n\in\mathbb{Z}}$ as follows. For $n\geq 0$, let 
    \begin{equation}
        \{\eta_n\}_{n\geq 0} = 1^201^30^2\cdots1^{k+1}0^k\cdots
    \end{equation}
    and extend it to the negative half-line by $\eta_{-n}=\eta_{n-1}$ for $n>0$. Then $\eta\in\Omega$, but $\eta\notin\Omega_{\mathscr{A}_{\textup{BBS}}}$.
\end{remark}
\begin{remark}
    In Proposition~\textup{\ref{prop:TOmegainOmega}}, only the conditions on $\mathscr{A}$ in \textup{\ref{enum:unique_final_state}} and \textup{\ref{enum:reset_word}} are required; \textup{\ref{enum:bijective}} is not necessary. However, we need \textup{\ref{enum:bijective}} for Theorem~\textup{\ref{thm:inv_meas}}. 
\end{remark}

\begin{theorem}\label{thm:inv_meas}
    Let $\mathscr{A}= (Q, \{0,1\}, \varphi, \psi)$ be an automaton satisfying \textup{\ref{enum:bijective}}, \textup{\ref{enum:unique_final_state}}, and \textup{\ref{enum:reset_word}}. Let $N\in \mathbb{Z}_{\geq 1}$ be an integer provided by condition \textup{\ref{enum:reset_word}}. Then, $\nu_M$ is an invariant measure for $T_\mathscr{A}$ if and only if for each $j\in \{0,1\}$, any $n\in \mathbb{Z}_{\geq 1}$, and any $\mathbf{s}\in \{0,1\}^n$, 
    \begin{equation}\label{eq:equiv_condition}
        \frac{\sum_{q\in Q} \nu_M(\eta_1\cdots\eta_{N+n+3} = \tilde{\psi}(0^{N+1}\mathbf{s}j^2;q)\mid \eta_1 = 0)}{\sum_{q\in Q}\nu_M(\eta_1\cdots\eta_{N+n+2} = \tilde{\psi}(0^{N+1}\mathbf{s}j;q)\mid \eta_1 = 0)} = 
        \begin{cases}
        p \quad \text{if}\ j = 0\\
        r \quad \text{if}\ j = 1.
        \end{cases}
    \end{equation}
\end{theorem}

\begin{proof}
    By Proposition~\ref{prop:Mar_dist_supp}, $\nu_M(\Omega)=1$. Moreover, by Proposition~\ref{prop:TOmegainOmega}, $T_\mathscr{A}(\Omega)\subset\Omega$. Hence
    \begin{equation}
        \nu_M(T_\mathscr{A}^{-1}(\Omega))=1.
    \end{equation}
    Since $T_\mathscr{A}$ commutes with shift operation, the distribution of $\tilde{\eta}=T_\mathscr{A}(\eta)$ is shift-invariant.
    
    We first derive a formula for the transition probabilities of a configuration $\tilde{\eta}$. This computation does not rely on either the invariance of $\nu_M$ or condition~\eqref{eq:equiv_condition}. Let $m\geq N+3$, $j\in\{0,1\}$, and
    $\mathbf{s}\in\{0,1\}^{m-N-2}$. We compute
    \begin{equation}
    \nu_M(\tilde{\eta}_{m+1}=j
    \mid
    \tilde{\eta}_1\cdots\tilde{\eta}_{m}
    =
    0^{N+1}\mathbf{s}j)
    =
    \frac{
        \sum_{q\in Q}
        \nu_M(
            \tilde{\eta}_1\cdots\tilde{\eta}_{m+1}
            =
            0^{N+1}\mathbf{s}j^2,\,
            Q_{m+1}=q
        )
    }{
        \sum_{q\in Q}
        \nu_M(
            \tilde{\eta}_1\cdots\tilde{\eta}_{m}
            =
            0^{N+1}\mathbf{s}j,\,
            Q_m=q
        )
    }.
    \label{eq:common_transition_1}
    \end{equation}
    Using condition~\textup{\ref{enum:reset_word}} together with the inverse transition map guaranteed by \textup{\ref{enum:bijective}}, we obtain
    \begin{equation}
    \frac{
        \sum_{q\in Q}
        \nu_M(
            \tilde{\eta}_1\cdots\tilde{\eta}_{m+1}
            =
            0^{N+1}\mathbf{s}j^2,\,
            Q_{m+1}=q
        )
    }{
        \sum_{q\in Q}
        \nu_M(
            \tilde{\eta}_1\cdots\tilde{\eta}_{m}
            =
            0^{N+1}\mathbf{s}j,\,
            Q_m=q
        )
    }
    =
    \frac{
        \sum_{q\in Q}
        \nu_M(
            Q_0=q^*,\,
            \eta_1\cdots\eta_{m+1}
            =
            \tilde{\psi}(0^{N+1}\mathbf{s}j^2;q)
        )
    }{
        \sum_{q\in Q}
        \nu_M(
            Q_0=q^*,\,
            \eta_1\cdots\eta_{m}
            =
            \tilde{\psi}(0^{N+1}\mathbf{s}j;q)
        )
    }.
    \label{eq:common_transition_2}
    \end{equation}
    We next rewrite the right-hand side of~\eqref{eq:common_transition_2}. For each $q\in Q$, the event in the numerator
    \begin{equation}
        \left\{
        Q_0=q^* \quad \text{and}\quad
        \eta_1\cdots\eta_{m+1}
        =
        \tilde{\psi}(0^{N+1}\mathbf{s}j^2;q)
        \right\}
    \end{equation}
    implies that $(Q_0,\eta_1)=(q^*,0)$. Indeed, by \textup{\ref{enum:reset_word}}, $Q_1=q^*$. Moreover, \textup{\ref{enum:bijective}} and \textup{\ref{enum:unique_final_state}} imply that
    \begin{equation}
        (Q_0,\eta_1)
        =
        (\tilde{\varphi}(Q_1, \tilde{\eta}_1), \tilde{\psi}(Q_1, \tilde{\eta}_1))
        =
        (\tilde{\varphi}(q^*,0),\tilde{\psi}(q^*,0))
        =
        (q^*,0).
    \end{equation}
    The same observation applies to the denominator. Hence,
    \begin{align}
         \frac{\sum_{q\in Q} \nu_M(Q_0 = q^*,\eta_1\cdots \eta_{m+1} = \tilde{\psi}(0^{N+1}\mathbf{s}j^2;q))}{\sum_{q\in Q} \nu_M(Q_0 = q^*,\eta_1\cdots\eta_{m} = \tilde{\psi}(0^{N+1}\mathbf{s}j;q))}
         &= \frac{\sum_{q\in Q}E_{\nu_M}[\nu_M(\eta_1\cdots \eta_{m+1} = \tilde{\psi}(0^{N+1}\mathbf{s}j^2;q)\mid\mathcal{F}_1);\{(Q_0,\eta_1)=(q^*,0)\}]}{\sum_{q\in Q}E_{\nu_M}[\nu_M(\eta_1\cdots\eta_{m} = \tilde{\psi}(0^{N+1}\mathbf{s}j;q)\mid \mathcal{F}_1);\{(Q_0,\eta_1) = (q^*,0)\}]} \\
         &= \frac{\sum_{q\in Q}\nu_M(\eta_1\cdots \eta_{m+1} = \tilde{\psi}(0^{N+1}\mathbf{s}j^2;q)\mid \eta_1=0)\nu_M((Q_0,\eta_1)=(q^*,0))}{\sum_{q\in Q}\nu_M(\eta_1\cdots\eta_{m} = \tilde{\psi}(0^{N+1}\mathbf{s}j;q)\mid \eta_1=0)\nu_M((Q_0,\eta_1) = (q^*,0))} \\
         &=\frac{\sum_{q\in Q}\nu_M(\eta_1\cdots\eta_{m+1} = \tilde{\psi}(0^{N+1}\mathbf{s}j^2;q)\mid \eta_1 = 0)}{\sum_{q\in Q}\nu_M(\eta_1\cdots\eta_{m} = \tilde{\psi}(0^{N+1}\mathbf{s}j;q)\mid \eta_1=0)}.
     \end{align}
    Therefore, we obtain
    \begin{equation}
    \nu_M(\tilde{\eta}_{m+1}=j
    \mid
    \tilde{\eta}_1\cdots\tilde{\eta}_{m}
    =
    0^{N+1}\mathbf{s}j)
    =\frac{\sum_{q\in Q}\nu_M(\eta_1\cdots\eta_{m+1} = \tilde{\psi}(0^{N+1}\mathbf{s}j^2;q)\mid \eta_1 = 0)}{\sum_{q\in Q}\nu_M(\eta_1\cdots\eta_{m} = \tilde{\psi}(0^{N+1}\mathbf{s}j;q)\mid \eta_1=0)}.
    \label{eq:if_part}
    \end{equation}

    We now prove the two implications. First, assume that condition~\eqref{eq:equiv_condition} holds. For a word $\mathbf{s}\in\{0,1\}^n$ appearing in~\eqref{eq:equiv_condition}, set $m=N+n+2$. Then $m\geq N+3$ and $\mathbf{s}\in\{0,1\}^{m-N-2}$. Thus,~\eqref{eq:if_part} coincides exactly with the left-hand side of~\eqref{eq:equiv_condition}. Therefore, 
    \begin{equation}
        \nu_M(\tilde{\eta}_{m+1} = j
        \mid
        \tilde{\eta}_1\cdots\tilde{\eta}_{m}
        =
        0^{N+1}\mathbf{s}j)
        =
        \begin{cases}
        p \quad \text{if } j = 0,\\
        r \quad \text{if } j = 1.
        \end{cases}
    \end{equation}
    Together with Lemma~\ref{lem:000}, the distribution of $\tilde{\eta}$ is the Markov distribution with the same transition matrix as $\nu_M$. Therefore, $\nu_M$ is an invariant measure for $T_{\mathscr{A}}$.

    Conversely, assume that $\nu_M$ is an invariant measure for $T_{\mathscr{A}}$. Then the process $\tilde{\eta}=T_{\mathscr{A}}(\eta)$ has distribution $\nu_M$. Let $n\in\mathbb{Z}_{\geq 1}$, $\mathbf{s}\in\{0,1\}^n$, and $j\in\{0,1\}$ be arbitrary, and set $m=N+n+2$. Then $m\geq N+3$ and $\mathbf{s}\in\{0,1\}^{m-N-2}$. Hence, we may apply~\eqref{eq:if_part}. In particular, its transition probabilities coincide with those of the original Markov chain:
    \begin{equation}
        \nu_M(\tilde{\eta}_{m+1} = j
        \mid
        \tilde{\eta}_1\cdots\tilde{\eta}_{m}
        =
        0^{N+1}\mathbf{s}j)
        =
        \begin{cases}
        p \quad \text{if } j = 0,\\
        r \quad \text{if } j = 1.
        \end{cases}
    \end{equation}
    Comparing this with~\eqref{eq:if_part} yields condition~\eqref{eq:equiv_condition}. This proves the converse implication.
\end{proof}

Theorem~\ref{thm:Ber_inv_meas} yields a sufficient condition for the Bernoulli distribution to be an invariant measure for $T_{\mathscr{A}}$.
 \begin{theorem}\label{thm:Ber_inv_meas}
     Suppose that an automaton $\mathscr{A}= (Q, \{0,1\}, \varphi, \psi)$ is particle-preserving and satisfies \textup{\ref{enum:bijective}}, \textup{\ref{enum:unique_final_state}}, and \textup{\ref{enum:reset_word}}. Then, the Bernoulli distribution $\nu_p$ is an invariant measure for $T_{\mathscr{A}}$.
 \end{theorem}
 \begin{proof}
      Combining Theorem~\ref{thm:inv_meas} with the fact that $\nu_p$ is a product measure, it is sufficient to prove that for each $j\in\{0,1\}$, any $n\in\mathbb{Z}_{\geq1}$, and any $\mathbf{s}\in\{0,1\}^n$,
     \begin{equation}\label{eq:target_Bernoulli}
         \frac{\sum_{q\in Q}\nu_p(\eta_1\cdots\eta_{N+n+3}=\tilde{\psi}(0^{N+1}\mathbf{s}j^2;q))}{\sum_{q\in Q}\nu_p(\eta_1\cdots\eta_{N+n+2}=\tilde{\psi}(0^{N+1}\mathbf{s}j;q))}
         =
         \begin{cases}
             p &\text{if}\ j=0\\
             1-p&\text{if}\ j=1
         \end{cases}
     \end{equation}
     holds. Here, $N$ is as defined in \ref{enum:reset_word}. By Remark~\ref{rem:particle_preserving}, since the automaton $\mathscr{A}$ is particle-preserving, we may assume without loss of generality that the associated weight function $w$ satisfies $w(0) = 0, w(1) = 1$, and $w(q^*)=0$. Using the weight function $w$, we can write the probability of $\eta_1\eta_2\cdots\eta_{n}=r_1r_2\cdots r_n$ under $\nu_p$ as follows. First, we have
     \begin{equation}
         \nu_p(\eta_1\eta_2\cdots\eta_n = r_1r_2\cdots r_n)= p^{n-S}(1-p)^{S}.
     \end{equation}
     Here, $S$ denotes the number of $1$'s in $r_1r_2\cdots r_n$, namely,
     \begin{equation}
         S = \sum_{k=1}^n w(r_k).
     \end{equation}
     Let $\eta_1\eta_2\cdots\eta_{n+N+2}=\tilde{\psi}(0^{N+1}\mathbf{s}j;q)$ for $j\in\{0,1\}$ and $\mathbf{s}=s_1s_2\cdots s_n\in\{0,1\}^n$. Note that $Q_0 = q^*$ by \ref{enum:reset_word}. Then, the particle-preserving property implies that
     \begin{equation}
         w(q^*) + \sum_{k=1}^{n+N+2}w(\eta_k) =(N+1)w(0) + \sum_{k=1}^n w(s_k) + w(j) + w(q)
     \end{equation}
     holds. We define $S(\mathbf{s},j)$ by
     \begin{equation}
         S(\mathbf{s},j) = \sum_{k=1}^nw(s_k) + w(j).
     \end{equation}
     Note that $S(\mathbf{s},j)$ does not depend on $q$. Hence, we obtain
     \begin{equation}
         \sum_{k=1}^{n+N+2}w(\eta_k) = S(\mathbf{s}, j) + w(q).
     \end{equation}
     Therefore,
     \begin{equation}
         \nu_p(\eta_1\cdots\eta_{N+n+2}=\tilde{\psi}(0^{N+1}\mathbf{s}j;q)) = [p^{n+N+2-S(\mathbf{s},j)}(1-p)^{S(\mathbf{s},j)}]\cdot p^{-w(q)}(1-p)^{w(q)}
         \label{eq:nup_1}
     \end{equation}
     holds. By the same argument, we obtain
     \begin{equation}
         \nu_p(\eta_1\cdots\eta_{N+n+3}=\tilde{\psi}(0^{N+1}\mathbf{s}j^2;q)) = [p^{n+N+3-S(\mathbf{s},j)}(1-p)^{S(\mathbf{s},j)}]\cdot p^{-w(q)-w(j)}(1-p)^{w(q)+w(j)}.
         \label{eq:nup_2}
     \end{equation}
     After canceling the common factors $p^{n+N+2-S(\mathbf{s},j)}(1-p)^{S(\mathbf{s},j)}$ in \eqref{eq:nup_1} and \eqref{eq:nup_2}, the left-hand side of \eqref{eq:target_Bernoulli} becomes
     \begin{align}
         \frac{\sum_{q\in Q}\nu_p(\eta_1\cdots\eta_{N+n+3}=\tilde{\psi}(0^{N+1}\mathbf{s}j^2;q))}{\sum_{q\in Q}\nu_p(\eta_1\cdots\eta_{N+n+2}=\tilde{\psi}(0^{N+1}\mathbf{s}j;q))}
         &= \frac{p^{1-w(j)}(1-p)^{w(j)}\sum_{q\in Q}p^{-w(q)}(1-p)^{w(q)}}{\sum_{q\in Q}p^{-w(q)}(1-p)^{w(q)}} \\
         &=p^{1-w(j)}(1-p)^{w(j)}\\
         &=
         \begin{cases}
             p &\text{if} \ j=0 \\
             1-p &\text{if} \ j= 1.
         \end{cases}
     \end{align}
 \end{proof}

In the following sections, we examine whether the Markov distribution $\nu_M$ is invariant for the three models introduced in Section~\ref{sec:models}. Even when $\nu_M$ is not invariant, the marginal distribution at each site remains unchanged for all three models. This is a consequence of particle preservation.
\begin{theorem}
    Let $\mathscr{A} = (Q, \{0,1\}, \varphi, \psi)$ be an automaton. Let $\mu$ be a probability measure on $\{0,1\}^\mathbb{Z}$ that satisfies the following assumptions:
    \begin{enumerate}
        \item $\mu$ is shift-invariant.
        \item $\mu(\Omega_\mathscr{A})=1$.
        \item $\mathscr{A}$ is particle-preserving.
    \end{enumerate}
    Then, the output measure for $\tilde{\eta} = T_{\mathscr{A}}(\eta)$ is also shift-invariant. Moreover, its marginal distribution is identical to that of the input measure $\mu$, that is,
    \begin{equation}
        \mu(\tilde{\eta}_n=1)= \mu(\eta_n=1).
    \end{equation}
\end{theorem}
\begin{proof}
    First, note that the map $T_\mathscr{A}$ commutes with the shift operator. Therefore, the shift-invariance of $\mu$ implies the shift-invariance of the output measure. Similarly, the distribution of the carrier-state process $Q_n$ is stationary. In fact, we have
    \begin{equation}
        Q_n = \lim_{k\to\infty}\varphi(q;\eta_{n-k} \cdots\eta_{n})\stackrel{d}{=}\lim_{k\to \infty}\varphi(q;\eta_{n-1-k}\cdots\eta_{n-1})=Q_{n-1}.
    \end{equation}

    By Remark~\ref{rem:particle_preserving}, we may assume that the weight function satisfies $w(0)=0$ and $w(1)=1$. Summing both sides of the particle-preserving relation \eqref{eq:conserv} over the interval $[1,L]$ and dividing by the length $L$, we obtain
    \begin{equation}\label{eq:sum_particle_preserving}
        \frac{1}{L}\sum_{n=1}^L(w(\tilde{\eta}_n)-w(\eta_n)) = \frac{w(Q_0)-w(Q_L)}{L}.
    \end{equation}
    Since $Q_0 \stackrel{d}{=} Q_L$, for any $\varepsilon>0$,
    \begin{equation}
        \lim_{L\to\infty}P\left(\frac{w(Q_L)}{L}>\varepsilon\right) = \lim_{L\to\infty}P\left(\frac{w(Q_0)}{L}>\varepsilon\right) = 0.
    \end{equation}
    Therefore, the right-hand side of~\eqref{eq:sum_particle_preserving} converges to $0$ in probability as $L\to\infty$. Since the absolute value of the left-hand side of~\eqref{eq:sum_particle_preserving} is bounded by $1$, this convergence is in fact in $\mathcal{L}^1$. Thus, we find
    \begin{align}
        |\mu(\tilde{\eta}_n = 1)- \mu(\eta_n = 1)| &= \left| E_\mu\left[ 
        \frac{1}{L}\sum_{n=1}^L(w(\tilde{\eta}_n)-w(\eta_n))\right] \right| \\
        &\leq E_\mu\left[ \left| \frac{w(Q_0) - w(Q_L)}{L}\right|\right] \\
        &\to 0 \quad (\text{as}\ L\to \infty).
    \end{align}
    This completes the proof.
\end{proof}

\section{Soliton Models Associated with 2-Letter, 3-State Mealy Automata}
\label{sec:models}

In this section, we introduce the three soliton models associated with $2$-letter, $3$-state Mealy automata classified in~\cite{maeno2025solitons}. These models provide concrete examples of automata satisfying the assumptions developed in the previous section. We first describe their dynamics both in terms of carrier rules and Mealy automata representations, and then verify that the corresponding time evolutions are well-defined on the configuration space $\Omega$.

\subsection{Model Description}

Among the $3$-state Mealy automata over a $2$-letter alphabet, those that are bijective, transitive, particle-preserving, and locally interacting fall into exactly three isomorphism classes, namely BBS-C(2), BBS-S(2), and BBS-V(2)~\cite{maeno2025solitons}: BBS-C(2) is the classical capacity-two BBS. BBS-S(2) admits a linearization. BBS-V(2) is equivalent to the ultradiscrete Lotka-Volterra equation. Moreover, all three exhibit solitonic behavior.

Here we describe each of the three models, BBS-C(2), BBS-S(2), and BBS-V(2). The first model, BBS-C(2), is a BBS with a finite carrier capacity of $2$. Thus, the carrier can hold at most two balls, and once it reaches capacity, it simply passes occupied sites without picking up additional balls. The second model is termed BBS-S(2). Its carrier has capacity $1$. When the carrier picks up a ball, it moves two sites ahead and drops the ball there if the site is empty. Otherwise, it again moves two sites ahead and repeats the same procedure. In BBS-V(2), the carrier has capacity $1$ and deposits a picked-up ball at the second nearest empty site. The configuration between the pickup and drop-off sites remains unchanged. Examples of the time evolution for each model are shown in Fig.~\ref{fig:time_evolution}.

\begin{figure}[htbp]
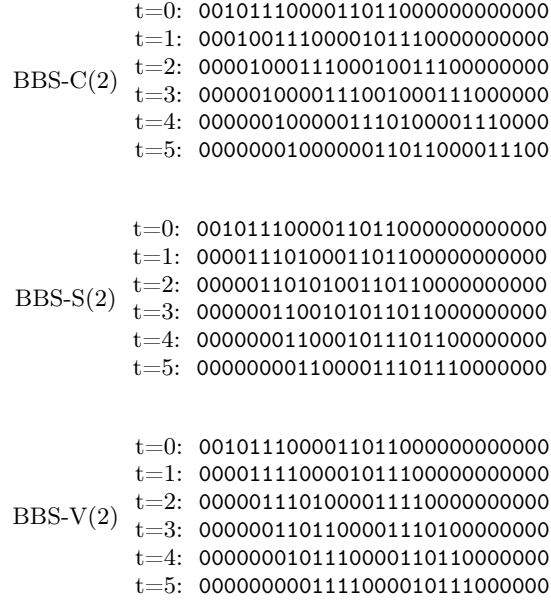

        \centering
        \[
        \text{BBS-C(2)} \
        \begin{array}{ll}
            \text{t=0:} & \texttt{0010111000011011000000000000} \\
            \text{t=1:} & \texttt{0001001110000101110000000000} \\
            \text{t=2:} & \texttt{0000100011100010011100000000} \\
            \text{t=3:} & \texttt{0000010000111001000111000000} \\
            \text{t=4:} & \texttt{0000001000001110100001110000} \\
            \text{t=5:} & \texttt{0000000100000011011000011100} \\
        \end{array}
        \] \\
        \[
        \text{BBS-S(2)} \
        \begin{array}{ll}
            \text{t=0:} & \texttt{0010111000011011000000000000} \\
            \text{t=1:} & \texttt{0000111010001101100000000000} \\
            \text{t=2:} & \texttt{0000011010100110110000000000} \\
            \text{t=3:} & \texttt{0000001100101011011000000000} \\
            \text{t=4:} & \texttt{0000000110001011101100000000} \\
            \text{t=5:} & \texttt{0000000011000011101110000000} \\
        \end{array}
        \] \\
        \[
        \text{BBS-V(2)} \
        \begin{array}{ll}
            \text{t=0:} & \texttt{0010111000011011000000000000} \\
            \text{t=1:} & \texttt{0000111100001011100000000000} \\
            \text{t=2:} & \texttt{0000011101000011110000000000} \\
            \text{t=3:} & \texttt{0000001101100001110100000000} \\
            \text{t=4:} & \texttt{0000000101110000110110000000} \\
            \text{t=5:} & \texttt{0000000001111000010111000000} \\
        \end{array}
        \]
        \caption{Examples of time evolutions for the three models, BBS-C(2), BBS-S(2), and BBS-V(2)}
        \label{fig:time_evolution}
\end{figure}

\subsection{Automaton Representations}

On the other hand, these models can be represented as Mealy automata. Let $Q = \{ q_0, q_1, q_2 \}$ and $S = \{ 0, 1 \}$.
Here, $s_k \in S$ corresponds to the number of balls at a given site $k$. Besides, $Q_k \in Q$ represents a carrier state right after passing a site $k$. The time evolution of each automaton is then described as follows: For BBS-C(2), the transition table is given in Table~\ref{tab:C2}. The table shows how each pair $(q,s)$ is updated by the automaton. That is, for each $(q,s)$, the corresponding entry gives $(\varphi(q,s), \psi(q,s))$. For example, $(q_0, 0)$ is mapped to $(\varphi(q_0,0), \psi(q_0,0))=(q_0, 0)$, $(q_1, 0)$ is mapped to $(\varphi(q_1,0), \psi(q_1,0))=(q_0, 1)$, and similarly for the other cases. In this case, $q_0$, $q_1$, and $q_2$ correspond respectively to the situations where the carrier holds $0$, $1$, and $2$ balls. In Fig.~\ref{fig:time_evolution_state_C}, we show the time evolution and the associated state sequences.

\begin{table}[htbp]
    \centering
    \caption{Transition table for BBS-C(2), giving $(\varphi(q,s), \psi(q,s))$}
    \label{tab:C2}
    \begin{tabular}{c|ccc}
        \hline
        \diagbox{$s$}{$q$} & $q_0$ & $q_1$ & $q_2$ \\
        \hline
        $0$ & $(q_0,0)$ & $(q_0,1)$ & $(q_1,1)$ \\
        $1$ & $(q_1,0)$ & $(q_2,0)$ & $(q_2,1)$ \\
        \hline
    \end{tabular}
\end{table}

\begin{figure}[htbp]
        \centering
        \small
        \[
        \text{BBS-C(2)}
        \quad
        \begin{array}{lll}
            \text{t=0:} & s_n & \texttt{0010111000011011000000000000} \\
            & \rd{Q_n} & \rd{\texttt{0010122100012122100000000000}} \\
            \text{t=1:} & s_n & \texttt{0001001110000101110000000000} \\
            & \rd{Q_n} & \rd{\texttt{0001001221000101221000000000}} \\
            \text{t=2:} & s_n & \texttt{0000100011100010011100000000} \\
            & \rd{Q_n} & \rd{\texttt{0000100012210010012210000000}} \\
            \text{t=3:} & s_n & \texttt{0000010000111001000111000000} \\
            & \rd{Q_n} & \rd{\texttt{0000010000122101000122100000}} \\
            \text{t=4:} & s_n & \texttt{0000001000001110100001110000} \\
            & \rd{Q_n} & \rd{\texttt{0000001000001221210001221000}} \\
            \text{t=5:} & s_n & \texttt{0000000100000011011000011100} \\
            & \rd{Q_n} & \rd{\texttt{0000000100000012122100012210}} \\
        \end{array}
        \]
        \caption{Example of the time evolution of BBS-C(2) together with the corresponding sequence of states. We display $s_n$ and $Q_n$ for $n=1,2,\ldots,28$ with $Q_0=q_0$ at each time step. For $Q_n$, we show the index $i$ of the state $q_i$ with $Q_n=q_i$ ($i=0,1,2$)}
        \label{fig:time_evolution_state_C}
\end{figure}

For BBS-S(2), the transition table is instead given in Table~\ref{tab:S2}. In this case, $q_0$ represents the state where the carrier is empty, $q_1$ the state where it has a ball but just skips to the next site, and $q_2$ the state where it has a ball, skips one site ahead, and then drops off the ball if the following site is empty. We show an example of the time evolution and the associated state sequences in Fig.~\ref{fig:time_evolution_state_S}.

\begin{table}[htbp]
    \centering
    \caption{Transition table for BBS-S(2), giving $(\varphi(q,s), \psi(q,s))$}
    \label{tab:S2}
    \begin{tabular}{c|ccc}
        \hline
        \diagbox{$s$}{$q$} & $q_0$ & $q_1$ & $q_2$ \\
        \hline
        $0$ & $(q_0,0)$ & $(q_2,0)$ & $(q_0,1)$ \\
        $1$ & $(q_1,0)$ & $(q_2,1)$ & $(q_1,1)$ \\
        \hline
    \end{tabular}
\end{table}

\begin{figure}[htbp]
        \centering
        \small
        \[
        \text{BBS-S(2)}
        \quad
        \begin{array}{lll}
            \text{t=0:} & s_n & \texttt{0010111000011011000000000000} \\
            & \rd{Q_n} & \rd{\texttt{0012121200012012000000000000}} \\
            \text{t=1:} & s_n & \texttt{0000111010001101100000000000} \\
            & \rd{Q_n} & \rd{\texttt{0000121212001201200000000000}} \\
            \text{t=2:} & s_n & \texttt{0000011010100110110000000000} \\
            & \rd{Q_n} & \rd{\texttt{0000012012120120120000000000}} \\
            \text{t=3:} & s_n & \texttt{0000001100101011011000000000} \\
            & \rd{Q_n} & \rd{\texttt{0000001200121212012000000000}} \\
            \text{t=4:} & s_n & \texttt{0000000110001011101100000000} \\
            & \rd{Q_n} & \rd{\texttt{0000000120001212121200000000}} \\
            \text{t=5:} & s_n & \texttt{0000000011000011101110000000} \\
            & \rd{Q_n} & \rd{\texttt{0000000012000012121212000000}} \\
        \end{array}
        \]
        \caption{Example of the time evolution of BBS-S(2) together with the corresponding sequence of states. We display $s_n$ and $Q_n$ for $n=1,2,\ldots,28$ with $Q_0=q_0$ at each time step. For $Q_n$, we show the index $i$ of the state $q_i$ with $Q_n=q_i$ ($i=0,1,2$)}
        \label{fig:time_evolution_state_S}
\end{figure}

Similarly, the transition table for BBS-V(2) is given in Table~\ref{tab:V2}. In this case, $q_0$ corresponds to the state where the carrier is empty, $q_1$ to the state where it has picked up a ball and is searching for the first empty site, and $q_2$ to the state where it is searching for the second empty site. Figure~\ref{fig:time_evolution_state_V} illustrates a time evolution and the associated state sequences.

\begin{table}[H]
    \centering
    \caption{Transition table for BBS-V(2), giving $(\varphi(q,s), \psi(q,s))$}
    \label{tab:V2}
    \begin{tabular}{c|ccc}
        \hline
        \diagbox{$s$}{$q$} & $q_0$ & $q_1$ & $q_2$ \\
        \hline
        $0$ & $(q_0,0)$ & $(q_2,0)$ & $(q_0,1)$ \\
        $1$ & $(q_1,0)$ & $(q_1,1)$ & $(q_2,1)$ \\
        \hline
    \end{tabular}
\end{table}

\begin{figure}[H]
        \centering
        \small
        \[
        \text{BBS-V(2)}
        \quad
        \begin{array}{lll}
            \text{t=0:} & s_n & \texttt{0010111000011011000000000000} \\
            & \rd{Q_n} & \rd{\texttt{0012222000011222000000000000}} \\
            \text{t=1:} & s_n & \texttt{0000111100001011100000000000} \\
            & \rd{Q_n} & \rd{\texttt{0000111120001222200000000000}} \\
            \text{t=2:} & s_n & \texttt{0000011101000011110000000000} \\
            & \rd{Q_n} & \rd{\texttt{0000011122000011112000000000}} \\
            \text{t=3:} & s_n & \texttt{0000001101100001110100000000} \\
            & \rd{Q_n} & \rd{\texttt{0000001122200001112200000000}} \\
            \text{t=4:} & s_n & \texttt{0000000101110000110110000000} \\
            & \rd{Q_n} & \rd{\texttt{0000000122220000112220000000}} \\
            \text{t=5:} & s_n & \texttt{0000000001111000010111000000} \\
            & \rd{Q_n} & \rd{\texttt{0000000001111200012222000000}} \\
        \end{array}
        \]
        \caption{Example of the time evolution of BBS-V(2) together with the corresponding sequence of states. We display $s_n$ and $Q_n$ for $n=1,2,\ldots,28$ with $Q_0=q_0$ at each time step. For $Q_n$, we show the index $i$ of the state $q_i$ with $Q_n=q_i$ ($i=0,1,2$)}
        \label{fig:time_evolution_state_V}
\end{figure}

From these tables, it is clear that the automata are bijective. 
Indeed, the transition tables for their inverses are given in Tables~\ref{tab:C2-inv}, 
\ref{tab:S2-inv}, and \ref{tab:V2-inv}.

\begin{table}[H]
    \centering
    \caption{Transition table for the inverse of BBS-C(2), giving $(\tilde{\varphi}(q,s),\tilde{\psi}(q,s))$}
    \label{tab:C2-inv}
    \begin{tabular}{c|ccc}
        \hline
        \diagbox{$s$}{$q$} & $q_0$ & $q_1$ & $q_2$ \\
        \hline
        $0$ & $(q_0,0)$ & $(q_0,1)$ & $(q_1,1)$ \\
        $1$ & $(q_1,0)$ & $(q_2,0)$ & $(q_2,1)$ \\
        \hline
    \end{tabular}
\end{table}

\begin{table}[H]
    \centering
    \caption{Transition table for the inverse of BBS-S(2), giving $(\tilde{\varphi}(q,s),\tilde{\psi}(q,s))$}
    \label{tab:S2-inv}
    \begin{tabular}{c|ccc}
        \hline
        \diagbox{$s$}{$q$} & $q_0$ & $q_1$ & $q_2$ \\
        \hline
        $0$ & $(q_0,0)$ & $(q_0,1)$ & $(q_1,0)$ \\
        $1$ & $(q_2,0)$ & $(q_2,1)$ & $(q_1,1)$ \\
        \hline
    \end{tabular}
\end{table}

\begin{table}[H]
    \centering
    \caption{Transition table for the inverse of BBS-V(2), giving $(\tilde{\varphi}(q,s),\tilde{\psi}(q,s))$}
    \label{tab:V2-inv}
    \begin{tabular}{c|ccc}
        \hline
        \diagbox{$s$}{$q$} & $q_0$ & $q_1$ & $q_2$ \\
        \hline
        $0$ & $(q_0,0)$ & $(q_0,1)$ & $(q_1,0)$ \\
        $1$ & $(q_2,0)$ & $(q_1,1)$ & $(q_2,1)$ \\
        \hline
    \end{tabular}
\end{table}

For BBS-C(2), the inverse coincides with the original one, while for BBS-S(2) and BBS-V(2), the inverse agrees with the original after interchanging the roles of $q_1$ and $q_2$. More precisely, in the inverse dynamics of BBS-S(2), $q_1$ represents the state where the carrier has a ball and checks every other site for an empty site at which to drop off the ball, whereas $q_2$ represents the intermediate state in which the carrier passes over a skipped site. In the inverse dynamics of BBS-V(2), $q_1$ represents the state where the carrier is searching for the second empty site, whereas $q_2$ represents the state where it has picked up a ball and is searching for the first empty site. Hence, in all three models, the original configuration can be reconstructed from the configuration after one time step by running the carrier from right to left according to the same pickup and drop-off rule.

\subsection{Well-Definedness on the Configuration Space}

Next, we check that the time evolutions $T_\mathscr{A}$ for BBS-C(2), BBS-S(2), and BBS-V(2) are well-defined on our configuration space $\Omega$ defined in~\eqref{def:Omega}.
\begin{proposition}\label{prop:preserveOmega}
    If $\mathscr{A}$ is the automaton corresponding to \textup{BBS-C(2)}, \textup{BBS-S(2)}, or \textup{BBS-V(2)}, it follows that $\Omega\subset \Omega_{\mathscr{A}}$ and $T_\mathscr{A}(\Omega)\subset \Omega$.
\end{proposition}
\begin{proof}
    From Tables~\ref{tab:C2}, \ref{tab:S2}, and \ref{tab:V2}, the automata of \textup{BBS-C(2)}, \textup{BBS-S(2)}, and \textup{BBS-V(2)} satisfy~\ref{enum:bijective}, \ref{enum:unique_final_state}, and \ref{enum:reset_word} with $q^* = q_0$ and $N = 2$.
    This proposition follows from Proposition~\ref{prop:TOmegainOmega}.
\end{proof}

\begin{remark}\label{rem:Omega_ext}
    For the \textup{BBS-C(2)} automaton $\mathscr{A} =\mathscr{A}_{\textup{{BBS-C(2)}}}$, $\Omega_{\textup{ext}}\subset \Omega_\mathscr{A}$ and $T_{\mathscr{A}}(\Omega_{\textup{ext}})\subset \Omega_{\textup{ext}}$ hold. Here,
    \begin{equation}
        \begin{aligned}
            \Omega_\textup{ext} &= \limsup_{m \to \infty} E_m \cap \limsup_{m \to -\infty} E_m \\
            E_m &= \{ \eta \in \{0, 1\}^{\mathbb{Z}} \mid \eta_m = \eta_{m + 1} = 0 \}.
        \end{aligned}
    \end{equation}
    Although the inclusion $\Omega_{\textup{ext}}\subset \Omega_\mathscr{A}$ also holds for the automata corresponding to $\textup{BBS-S(2)}$ and $\textup{BBS-V(2)}$, the space $\Omega_{\textup{ext}}$ is not invariant under their time evolutions.
\end{remark}
\begin{proof}
See Appendix~\ref{app:Omega_ext}.
\end{proof}

\section{Soliton Dynamics of BBS-C(2), BBS-S(2), and BBS-V(2)}
\label{sec:solitons}

The three models introduced in~\cite{maeno2025solitons} exhibit solitonic behavior. It was shown in~\cite{maeno2025solitons} that every initial configuration admits a decomposition into solitons. In this section, we review the properties of solitons (types, bare velocities, and phase shifts) in these models.

\subsection{Solitons and Their Velocities}

For BBS-C(2), a free soliton of size $k \in \mathbb{Z}_{\geq 2}$ is represented by $``1^k00"$. Here, $1^k$ denotes a block of $k$ consecutive occupied sites. The velocity of such a soliton is $v = 2$. There is also a size-$1$ free soliton $``10"$ or $``01"$, whose velocity is $1$.

We now introduce a soliton decomposition for BBS-C(2). Let $\eta\in \Omega$. We first decompose the configuration into components separated by occurrences of $``00"$. Each component is then of the form
\begin{equation}
    1^{a_1}01^{a_2}0\cdots01^{a_r},\qquad (a_1,a_2,\ldots,a_r\geq 1).
\end{equation}
We apply the following procedure to each component. If $r=1$, then the component is $1^{a_1}$. If $a_1=1$, we regard $10$ as a $1$-soliton. If $a_1\geq 2$, we regard $1^{a_1}00$ as a single soliton. Suppose $r\geq 2$. Consider the initial part
\begin{equation}
    1^{a_1}01^{a_2}
\end{equation}
of the component. If $a_1=1$, then we remove the pair $10$ consisting of the initial $1$ and the following $0$. This removed pair is defined to be a $1$-soliton. The remaining component is 
\begin{equation}
    1^{a_2}01^{a_3}0\cdots 01^{a_r}.
\end{equation}
If $a_1\geq 2$, then we remove the pair $01$ consisting of the $0$ immediately after the first $1$-block and the leftmost $1$ of the second $1$-block. That is, we remove the underlined part in 
\begin{equation}
    1^{a_1}\underline{01}1^{a_2-1}.
\end{equation}
The removed pair is defined to be a $1$-soliton. The remaining component is 
\begin{equation}
    1^{a_1+a_2-1}01^{a_3}0\cdots01^{a_r}.
\end{equation}
We repeat this procedure on the remaining component until it becomes a single block $1^k$. Finally, if $k=1$, we regard $10$ as a $1$-soliton, while if $k\geq 2$, we regard $1^k00$ as one soliton. In Fig.~\ref{fig:soliton_C(2)}, we show examples of these solitons and their time evolutions.

\begin{figure}[H]
        \centering
        \[
        ``10" \
        \begin{array}{ll}
            \text{t=0:} & \texttt{00{\rd\su{10}}0000000} \\
            \text{t=1:} & \texttt{000{\rd\su{10}}000000} \\
            \text{t=2:} & \texttt{0000{\rd\su{10}}00000} \\
        \end{array}
        \] \\
        \[
        ``1^200" \
        \begin{array}{ll}
            \text{t=0:} & \texttt{00{\rd\su{1100}}00000} \\
            \text{t=1:} & \texttt{0000{\rd\su{1100}}000} \\
            \text{t=2:} & \texttt{000000{\rd\su{1100}}0} \\
        \end{array}
        \] \\
        \[
        ``1^300" \
        \begin{array}{ll}
            \text{t=0:} & \texttt{00{\rd\su{11100}}0000} \\
            \text{t=1:} & \texttt{0000{\rd\su{11100}}00} \\
            \text{t=2:} & \texttt{000000{\rd\su{11100}}} \\
        \end{array}
        \]
        \caption{Examples of solitons in BBS-C(2). Solitons are underlined and highlighted in color}
        \label{fig:soliton_C(2)}
\end{figure}

BBS-S(2), on the other hand, admits only two types of free solitons, namely $``11"$ and $``10"$, which correspond to the two types of solitary waves described in~\cite{maeno2025solitons}. In this model, the soliton decomposition can be obtained by viewing a block $1^{2k}$ as $k$ copies of $``11"$, and a block $1^{2k+1}$ as $k$ copies of $``11"$ plus one $``10"$. Their velocities are $v_{11} = 1$ for $``11"$ and $v_{10} = 2$ for $``10"$ as illustrated in Fig.~\ref{fig:soliton_S(2)}.

\begin{figure}[htbp]
        \centering
        \[
        ``11" \
        \begin{array}{ll}
            \text{t=0:} & \texttt{00{\rd\su{11}}0000000} \\
            \text{t=1:} & \texttt{000{\rd\su{11}}000000} \\
            \text{t=2:} & \texttt{0000{\rd\su{11}}00000} \\
        \end{array}
        \] \\
        \[
        ``11"+``11" \
        \begin{array}{ll}
            \text{t=0:} & \texttt{00{\rd\su{11}}{\bl\du{11}}00000} \\
            \text{t=1:} & \texttt{000{\rd\su{11}}{\bl\du{11}}0000} \\
            \text{t=2:} & \texttt{0000{\rd\su{11}}{\bl\du{11}}000} \\
        \end{array}
        \] \\
        \[
        ``10" \
        \begin{array}{ll}
            \text{t=0:} & \texttt{00{\rd\su{10}}0000000} \\
            \text{t=1:} & \texttt{0000{\rd\su{10}}00000} \\
            \text{t=2:} & \texttt{000000{\rd\su{10}}000} \\
        \end{array}
        \] \\
        \[
        ``10"+``10" \
        \begin{array}{ll}
            \text{t=0:} & \texttt{00{\rd\su{10}}{\bl\du{10}}00000} \\
            \text{t=1:} & \texttt{0000{\rd\su{10}}{\bl\du{10}}000} \\
            \text{t=2:} & \texttt{000000{\rd\su{10}}{\bl\du{10}}0} \\
        \end{array}
        \]
        \caption{Examples of solitons in BBS-S(2). Different solitons are distinguished by single and double underlining and by different colors}
        \label{fig:soliton_S(2)}
\end{figure}

For BBS-V(2), we first specify the shape of a soliton. For $m\geq1$, we regard the word $1^m0$ as a phase of an $m$-soliton. In addition, for $m\geq 2$, we also regard the words $1^{m-n}01^{n}0$, where $1\leq n \leq m-1$, as different phases of the same $m$-soliton. Thus, the unique phase of a $1$-soliton is $10$. In all cases, the soliton contains exactly $m$ balls. The reason for including the second type is that under the time evolution of BBS-V(2), an isolated block $1^m0$ does not simply remain a consecutive block of $1$'s (see Fig.~\ref{fig:soliton_V(2)}). Accordingly, the definition of a soliton used here is slightly different from that in~\cite{maeno2025solitons}.

\begin{figure}[H]
        \centering
        \[
        m=1 \
        \begin{array}{ll}
            \text{t=0:} & \texttt{00{\rd\su{10}}0000000} \\
            \text{t=1:} & \texttt{0000{\rd\su{10}}00000} \\
            \text{t=2:} & \texttt{000000{\rd\su{10}}000} \\
        \end{array}
        \] \\
        \[
        m=2 \
        \begin{array}{ll}
            \text{t=0:} & \texttt{00{\rd\su{110}}000000} \\
            \text{t=1:} & \texttt{000{\rd\su{1010}}0000} \\
            \text{t=2:} & \texttt{00000{\rd\su{110}}000} \\
        \end{array}
        \] \\
        \[
        m=3 \
        \begin{array}{ll}
            \text{t=0:} & \texttt{00{\rd\su{1110}}00000} \\
            \text{t=1:} & \texttt{000{\rd\su{11010}}000} \\
            \text{t=2:} & \texttt{0000{\rd\su{10110}}00} \\
            \text{t=3:} & \texttt{000000{\rd\su{1110}}0} \\
        \end{array}
        \]
        \caption{Examples of solitons of size $m$ in BBS-V(2). Notice the periodic evolution of the soliton shape}
        \label{fig:soliton_V(2)}
\end{figure}

We now introduce a soliton decomposition for BBS-V(2). The following algorithm provides a decomposition of a finite configuration into solitons at each time step. Let $\eta\in\Omega$. We first decompose $\eta$ into maximal non-vacuum components separated by the reset word $``00"$. Since two groups of balls separated by $``00"$ do not interact during one time step, the soliton decomposition is defined independently for each such component. Consider one component of the form
\begin{equation}
    1^{a_1}01^{a_2}0\cdots01^{a_r}
\end{equation}
with $a_1,a_2,\ldots,a_r\geq 1$. We associate with this component the block-length sequence
\begin{equation}
    (a_1,a_2,\ldots,a_r).
\end{equation}
The soliton decomposition of this component is defined by the following recursive procedure. 

If $r=1$, then $1^{a_1}0$ is the unique soliton contained in this component. Otherwise, choose a maximal element $a_i$ of the sequence $(a_1,a_2,\ldots,a_r)$. If there is more than one maximal element, we choose the leftmost one. If the chosen $a_i$ has two neighbors (i.e., $1<i<r$), we pair $a_i$ with the larger of $a_{i-1}$ and $a_{i+1}$. If the two neighboring blocks have the same length, we choose the left neighbor. If $i=r$ or $i=1$, we pair it with the unique neighbor. The sum of the two paired numbers is recorded as the size of one soliton. The two paired blocks are then removed from the sequence, and the same procedure is applied to the remaining sequence. If a single block remains at the end, it is recorded as one soliton.
\begin{example}
    Consider a configuration 
    \begin{equation}
        0^21^101^201^101^301^201^10^2.
    \end{equation}
    The associated block-length sequence is $(1,2,1,3,2,1)$. The leftmost maximal block is $1^3$. Since its two neighboring blocks have lengths $1$ and $2$, we pair $1^3$ and $1^2$. This gives a soliton of size $3+2=5$. After removing these two blocks, the remaining block-length sequence is $(1,2,1,1)$. The leftmost maximal block is now $1^2$. Since the two neighbors have the same length $1$, we pair $1^2$ with its left neighbor $1^1$, obtaining a soliton of size $2+1=3$. Finally, the remaining two blocks are paired and give a soliton of size $1+1=2$. Therefore, the soliton content of this component is $\{5,3,2\}$. The soliton decomposition of this configuration and its time evolution are shown in Fig.~\ref{fig:decomp_alg} (see at $t=0$ for example).
\end{example}

\begin{figure}
    \begin{tabular}{ll}
    t=-2:  & \texttt{{\gr\su{10}}{\bl\du{10110}}{\gr\su{10}}{\rd\tu{1110110}}000000000000000000000} \\
    t=-1:  & \texttt{00{\bl\du{11010}}{\gr\su{110}}{\rd\tu{1101110}}00000000000000000000} \\
    t=0:  & \texttt{000{\bl\du{10110}}{\gr\su{10}}{\rd\tu{1110110}}{\gr\su{10}}000000000000000000} \\
    t=1:  & \texttt{00000{\bl\du{1110}}0{\rd\tu{1111010}}{\gr\su{110}}00000000000000000} \\
    t=2:  & \texttt{000000{\bl\du{11010}}{\rd\tu{1110110}}{\gr\su{1010}}000000000000000} \\
    t=3:  & \texttt{0000000{\bl\du{10110}}{\rd\tu{1101110}}0{\gr\su{110}}00000000000000} \\
    t=4:  & \texttt{000000000{\bl\du{1110}}{\rd\tu{1011110}}0{\gr\su{1010}}000000000000} \\
    t=5:  & \texttt{0000000000{\bl\du{110}}{\rd\tu{1101110}}{\bl\du{10}}0{\gr\su{110}}00000000000} \\
    t=6:  & \texttt{00000000000{\bl\du{10}}{\rd\tu{1110110}}{\bl\du{110}}0{\gr\su{1010}}000000000} \\
    t=7: & \texttt{0000000000000{\rd\tu{1111010}}{\bl\du{1110}}00{\gr\su{110}}00000000} \\
    t=8: & \texttt{00000000000000{\rd\tu{1110110}}{\bl\du{11010}}0{\gr\su{1010}}000000} \\
    t=9: & \texttt{000000000000000{\rd\tu{1101110}}{\bl\du{10110}}00{\gr\su{110}}00000} \\
    t=10: & \texttt{0000000000000000{\rd\tu{1011110}}0{\bl\du{1110}}00{\gr\su{1010}}000} \\
    t=11: & \texttt{000000000000000000{\rd\tu{111110}}0{\bl\du{11010}}00{\gr\su{110}}00} \\
    t=12: & \texttt{0000000000000000000{\rd\tu{1111010}}{\bl\du{10110}}00{\gr\su{1010}}}
    \end{tabular}
    \caption{Time evolution of the configuration \texttt{101011010111011} under BBS-V(2). The soliton decomposition given by the peak-pairing algorithm is indicated by single (green), double (blue), and triple (red) underlining corresponding to the $2$-, $3$-, and $5$-solitons, respectively. When a soliton is isolated, its shape is of the form $1^m0$ or $1^{m-n}01^n0$. The right ends of the configurations are aligned by padding with zeros}
    \label{fig:decomp_alg}
\end{figure}

We next define the position of each soliton. The algorithm gives not only the size of each soliton but also the set of balls belonging to it. More precisely, if two blocks are paired in the algorithm, then all balls contained in these two blocks are assigned to the same soliton. If one block remains unpaired at the end, then the balls in this block form one soliton. Fix $\eta\in \Omega$, and consider a soliton $\gamma$ in $\eta$ obtained by the algorithm. We write $m_\gamma$ for its size, and denote by $P_\gamma\subset \mathbb{Z}$ the set of positions of the balls belonging to $\gamma$. By definition, $|P_\gamma|=m_\gamma$. We define the position of the soliton $\gamma$ as the average of the positions of the balls in $P_\gamma$, i.e., its center of mass:
\begin{equation}
    X_\gamma(\eta) = \frac{1}{m_\gamma}\sum_{x\in P_\gamma}x.
\end{equation}
Although the shape of the soliton changes with time, the average position of its balls moves linearly. More precisely, for an isolated $m$-soliton, its position moves to the right by $(m+1)/m$ at each time step (Fig.~\ref{fig:soliton_V(2)}). Thus the position defined above yields the natural velocity 
\begin{equation}
    v_m = \frac{m+1}{m}.
\end{equation}

\subsection{Phase Shifts}

When a faster soliton is initially located behind a slower one, it eventually overtakes the slower soliton. Note that solitons always propagate to the right, i.e., $v > 0$. However, interactions between the solitons during the overtaking process shift their positions compared with the non-interacting evolution. This displacement is called a phase shift. In the GHD of soliton systems, the phase shifts together with the soliton velocities in the absence of interactions, often called the bare velocities, determine the macroscopic behavior of the system.

In this section, we determine the phase shifts for the three models. Note that when a soliton is already ahead of another soliton with a smaller or equal velocity, then the two solitons never interact afterwards. We denote by $\Phi(J, K)$ the phase shift of a soliton $J$ resulting from its interaction with another soliton $K$. For example, in the original BBS (Example~\ref{ex:TakahashiSatsuma}), the phase shift is given by
\begin{equation}
    \Phi(1^k0^k,1^l0^l)
    =\begin{cases}
        2(k\wedge l), &k>l \\
        -2(k\wedge l), &k<l,
    \end{cases}
    \label{eq:BBS_phase_shift}
\end{equation}
while solitons with the same velocity do not overtake each other.

For BBS-C(2), we have only two possible velocities: $v = 1$ for $``10"$ and $v=2$ for $``1^k00"$ with $k\geq2$. It is known~\cite{DaisukeTakahashi_1997} that
\begin{align}
    \Phi(1^k00, 10) = -\Phi(10, 1^k00) = 2
\end{align}
for $k \geq 2$. The phase shift is antisymmetric, $\Phi(J, K) = -\Phi(K, J)$, as in the original BBS~\eqref{eq:BBS_phase_shift}. Examples of interactions are shown in Fig.~\ref{fig:phase_C(2)}.

\begin{figure}[H]
        \centering
        \[
        ``1^200"+``10" \
        \begin{array}{ll}
            \text{t=0:} & \texttt{00{\rd\du{1100}}{\bl\su{10}}000000} \\
            \text{t=1:} & \texttt{0000{\rd\du{11}}{\bl\su{01}}{\rd\du{00}}0000} \\
            \text{t=2:} & \texttt{000000{\bl\su{10}}{\rd\du{1100}}00} \\
            \text{t=3:} & \texttt{0000000{\bl\su{10}}0{\rd\du{1100}}} \\
        \end{array}
        \] \\
        \[
        ``1^300"+``10" \
        \begin{array}{ll}
            \text{t=0:} & \texttt{00{\rd\du{11100}}{\bl\su{10}}00000000}\\
            \text{t=1:} & \texttt{0000{\rd\du{111}}{\bl\su{01}}{\rd\du{00}}000000} \\
            \text{t=2:} & \texttt{000000{\rd\du{11}}{\bl\su{01}}{\rd\du{100}}0000} \\
            \text{t=3:} & \texttt{00000000{\bl\su{10}}{\rd\du{11100}}00} \\
            \text{t=4:} & \texttt{000000000{\bl\su{10}}0{\rd\du{11100}}} \\
        \end{array}
        \]
        \caption{Examples of soliton interactions in BBS-C(2). Different solitons are distinguished by single and double underlining and by different colors}
        \label{fig:phase_C(2)}
\end{figure}

For BBS-S(2), we instead have
\begin{align}
    \Phi(11,10)=-1, \quad \Phi(10,11)=2,
\end{align}
as illustrated in Fig.~\ref{fig:phase_S(2)}. In this case, the phase shift is not antisymmetric with respect to the two arguments of $\Phi(\cdot, \cdot)$, namely
\begin{align}
    \Phi(J,K)\neq-\Phi(K,J).
\end{align}

\begin{figure}[htbp]
        \centering
        \[
        ``10"+``11" \
        \begin{array}{ll}
            \text{t=0:} & \texttt{00{\bl\su{10}}{\rd\du{11}}000} \\
            \text{t=1:} & \texttt{0000{\rd\du{11}}{\bl\su{10}}0} \\
        \end{array}
        \] \\
        \[
        ``10"+``10"+``11" \
        \begin{array}{ll}
            \text{t=0:} & \texttt{00{\gr\su{10}}{\bl\du{10}}{\rd\tu{11}}00000} \\
            \text{t=1:} & \texttt{0000{\gr\su{10}}{\rd\tu{11}}{\bl\du{10}}000} \\
            \text{t=2:} & \texttt{000000{\rd\tu{11}}{\gr\su{10}}{\bl\du{10}}0} \\
        \end{array}
        \] \\
        \[
        ``10"+``11"+``11" \
        \begin{array}{ll}
            \text{t=0:} & \texttt{00{\gr\su{10}}{\bl\du{11}}{\rd\tu{11}}000} \\
            \text{t=1:} & \texttt{0000{\bl\du{11}}{\rd\tu{11}}{\gr\su{10}}0} \\
        \end{array}
        \]
        \caption{Examples of soliton interactions in BBS-S(2). In the second example, a $``11"$-soliton interacts with two $``10"$-solitons. In the third example, two $``11"$-solitons interact with a $``10"$-soliton}
        \label{fig:phase_S(2)}
\end{figure}

Finally, for BBS-V(2), let $\gamma_m$ be an $m$-soliton, that is, a soliton consisting of $m$ balls of the form $1^m0$ or $1^{m-n}01^n0$ for some $1\leq n\leq m-1$. The phase shifts are obtained as 
\begin{align}
    \Phi(\gamma_l,\gamma_m) &= -m \left(1+\frac{2}{l}\right) \\
    \Phi(\gamma_m,\gamma_l) &= l+2
\end{align}
for $l>m\geq1$, as illustrated in Fig.~\ref{fig:phase_V(2)}. As in BBS-S(2), the phase shifts are not antisymmetric.

\begin{figure}[H]
        \centering
        \[
        (l,m)=(2,1) \
        \begin{array}{ll}
            \text{t=0:} & \texttt{00{\bl\su{10}}0{\rd\du{110}}00000} \\
            \text{t=1:} & \texttt{0000{\rd\du{1010}}{\bl\su{10}}000} \\
            \text{t=2:} & \texttt{000000{\rd\du{110}}0{\bl\su{10}}0} \\
        \end{array}
        \] \\
        \[
        (l,m)=(3,1) \
        \begin{array}{ll}
            \text{t=0:} & \texttt{00{\bl\su{10}}0{\rd\du{1110}}00000} \\
            \text{t=1:} & \texttt{0000{\rd\du{10110}}{\bl\su{10}}000} \\
            \text{t=2:} & \texttt{000000{\rd\du{1110}}0{\bl\su{10}}0} \\
        \end{array}
        \] \\
        \[
        (l,m)=(3,2) \
        \begin{array}{ll}
            \text{t=0:} & \texttt{00{\bl\su{110}}0{\rd\du{1110}}0000000} \\
            \text{t=1:} & \texttt{000{\bl\su{10}}{\rd\du{10110}}{\bl\su{10}}00000} \\
            \text{t=2:} & \texttt{00000{\rd\du{11010}}{\bl\su{110}}0000} \\
            \text{t=3:} & \texttt{000000{\rd\du{10110}}{\bl\su{1010}}00} \\
            \text{t=4:} & \texttt{00000000{\rd\du{1110}}0{\bl\su{110}}0} \\
        \end{array}
        \]
        \caption{Examples of interactions between an $l$-soliton (red, double-underlined) and an $m$-soliton (blue, single-underlined) in BBS-V(2)}
        \label{fig:phase_V(2)}
\end{figure}

The soliton velocities and phase shifts in the three models are summarized in Table~\ref{tab:prop}.

\begin{table}[htbp]
    \caption{Soliton types, velocities, and phase shifts in the three models}
    \label{tab:prop}
    \centering
  \begin{tabular}{c c c c c c}
    \toprule
     & \multicolumn{2}{c}{BBS-C(2)} & \multicolumn{2}{c}{BBS-S(2)} & BBS-V(2) \\
    \cmidrule(lr){2-3} \cmidrule(lr){4-5} \cmidrule(lr){6-6}
    Soliton & $1^k 00$ &$10$ or $01$ & $11$ & $10$ & $1^{m-n} 01^{n}0 \ (1\leq n \leq m-1)$ or $1^m0$ \\
    Bare velocity & $2$ &$1$ & $1$ & $2$ & $\frac{m+1}{m}$ \\
    \multirow{3}{*}{Phase shift} & \multicolumn{2}{c}{$\Phi(1^k00,10)=2,$} & \multicolumn{2}{c}{$\Phi(11,10)=-1,$} & $\Phi(\gamma_l, \gamma_m)=-m(1+\frac{2}{l}),$ \\
     & \multicolumn{2}{c}{$\Phi(10,1^k00)=-2$} & \multicolumn{2}{c}{$\Phi(10,11)=2$} & $\Phi(\gamma_m,\gamma_l)=l+2$ \\
     & \multicolumn{2}{c}{$(k\geq 2)$} & \multicolumn{2}{c}{} & $(l>m\geq1)$ \\
    \bottomrule
  \end{tabular}
\end{table}

\section{Invariant Measures for BBS-C(2), BBS-S(2), and BBS-V(2)}
\label{sec:measures}

In this section, we prove that Bernoulli product measures are invariant for BBS-S(2) and BBS-V(2), while two-sided space-homogeneous Markov distributions are invariant for BBS-C(2).

\subsection{Invariant Measures for BBS-C(2)}

\begin{theorem}\label{thm:inv_meas_C2}
    For any $p,r\in(0,1)$, a stationary Markov distribution $\nu_M$ is an invariant measure for the system \textup{BBS-C(2)}.
\end{theorem}
\begin{proof}
    Let $\tilde{\eta}$ denote the configuration obtained from $\eta$ after one time step of evolution. Note that this time-evolution map is shift-invariant. By Lemma~\ref{lem:000}, to prove this theorem, we only need to show that for any integer $n\geq 5$ and for any sequence $\tilde{\mathbf{s}} = (\tilde{s}_4,\ldots, \tilde{s}_{n-1})\in\{0,1\}^{n-4}$, the following equations hold:
    \begin{align}
        \nu_M(\tilde{\eta}_{n+1} = 0 \mid \tilde{\eta}_1\ldots\tilde{\eta}_n=0^3\tilde{\mathbf{s}}0) &= p \label{eq:C2Markovp}\\
        \nu_M(\tilde{\eta}_{n+1} = 1 \mid \tilde{\eta}_1\ldots\tilde{\eta}_n = 0^3\tilde{\mathbf{s}}1) &= r. \label{eq:C2Markovr}
    \end{align}
    
    We first prove~\eqref{eq:C2Markovp}. Let $F_k$ denote the event $\{Q_n = q_k, \tilde{\eta}_1\ldots\tilde{\eta}_n = 0^3\tilde{\mathbf{s}}0\}$ for $k=0, 1, 2$. Using the transition rules of BBS-C(2) and Lemma~\ref{lem:000}, it is sufficient to prove
    \begin{equation}\label{targetC(2)_00}
        \frac{\nu_M(F_0)+r \nu_M(F_1)}{\nu_M(F_0)+\nu_M(F_1)+\nu_M(F_2)} = p.
    \end{equation}
    The probability of each event, $\nu_M(F_k)$, can be computed by reconstructing the preimage configuration $\eta_1\ldots \eta_n$ where $\eta = T_{\text{BBS-C(2)}}^{-1}(\tilde{\eta})$ together with its initial state $Q_0$. When considering the ratio in \eqref{targetC(2)_00}, any part of $\eta$ that is identical regardless of the state $Q_n = q_k$ corresponds to a common factor in the probabilities, which will cancel out. Therefore, we only need to analyze the parts of $\eta$ that differ. Note that for BBS-C(2), recovering $Q_{n-1}$ and $\eta_n$ from $Q_n$ and $\tilde{\eta}_n$ is equivalent to moving the carrier from right to left.
    
    The possible preimage configurations $\eta$ fall into three cases. The first case is determined by the values of $\tilde{\eta}_{n-1}$ and $\tilde{\eta}_n$ at the end of the configuration. The remaining two cases arise from the following observation. Starting from any state $Q_n\in\{q_0,q_1,q_2\}$, the states are updated as follows: when the carrier passes $``00"$ in $\tilde{\eta}$ from right to left, the state becomes $q_0$ regardless of the initial state $Q_n$. Similarly, when it passes $``11"$, the state becomes $q_2$ regardless of $Q_n$.

    \emph{Case 1:} The output history ends with $\tilde{\eta}_{n-1} \tilde{\eta}_n = 0^2$.\\
    In this case, the tails, i.e., $(\eta_{n-2}, \eta_{n-1}, \eta_n)$, of the three configurations corresponding to $Q_n=q_0,q_1,q_2$ are shown in Table~\ref{tab:0case1}.

    \begin{table}[htbp]
    \caption{Transition probabilities and patterns (Case 1)}\label{tab:0case1}
    \centering
    \begin{tabular}{rcc}
        \toprule
        $Q_n$ & $\eta$ & Probability/Common factor \\
        \midrule
        $q_0$ & $\cdots\texttt{000}$ & $p^2$ \\
        $q_1$ & $\cdots\texttt{001}$ & $p(1-p)$ \\
        $q_2$ & $\cdots\texttt{011}$ & $r(1-p)$ \\
        \bottomrule
    \end{tabular}
    \end{table}

    \emph{Case 2:} $\tilde{\eta}$ ends with $\tilde{\eta}_{n-1}\tilde{\eta}_n=10$, and when the carrier runs from right to left, it encounters $``00"$ before $``11"$ in $\tilde{\eta}$. \\
    In this case, the three configurations share a common intermediate structure, differing only at the boundaries (Table~\ref{tab:0case2}).

    \begin{table}[htbp]
    \caption{Transition probabilities and patterns (Case 2)}\label{tab:0case2}
    \centering
    \begin{tabular}{rcc}
        \toprule
        $Q_n$ & $\eta$ & Probability/Common factor \\
        \midrule
        $q_0$ & $\texttt{001}\cdots\texttt{00}$ & $p^2(1-p)$ \\
        $q_1$ & $\texttt{001}\cdots\texttt{01}$ & $p(1-p)^2$ \\
        $q_2$ & $\texttt{011}\cdots\texttt{01}$ & $r(1-p)^2$ \\
        \bottomrule
    \end{tabular}
    \end{table}
    \emph{Case 3:} $\tilde{\eta}$ ends with $\tilde{\eta}_{n-1}\tilde{\eta}_n=10$, and when the carrier runs from right to left, it encounters $``11"$ before $``00"$ in $\tilde{\eta}$.\\
    Similarly to Case $2$, the three configurations share a common structure in the middle (Table~\ref{tab:0case3}).

    \begin{table}[htbp]
    \caption{Transition probabilities and patterns (Case 3)}\label{tab:0case3}
    \centering
    \begin{tabular}{rcc}
        \toprule
        $Q_n$ & $\eta$ & Probability/Common factor \\
        \midrule
        $q_0$ & $\mathtt{100}\cdots\mathtt{00}$ & $p^2(1-r)$ \\
        $q_1$ & $\mathtt{100}\cdots\mathtt{01}$ & $p(1-p)(1-r)$ \\
        $q_2$ & $\mathtt{110}\cdots\mathtt{01}$ & $(1-p)r(1-r)$ \\
        \bottomrule
    \end{tabular}
    \end{table}

    A key observation is that the ratio of the probabilities $\nu_M(F_0) : \nu_M(F_1) : \nu_M(F_2)$ is given by $p^2 : p(1-p) : r(1-p)$ in all three cases after factoring out a common term. Substituting this constant ratio into~\eqref{targetC(2)_00} confirms that the equality holds. This completes the proof for the $0\to0$ transition \eqref{eq:C2Markovp}.

    We next prove \eqref{eq:C2Markovr}. Similarly, let $F_k'$ denote the event $\{Q_n = q_k, \tilde{\eta}_1\ldots\tilde{\eta}_n=0^3\tilde{\mathbf{s}}1\}$ for $k=0, 1, 2$. It suffices to show that
    \begin{equation}\label{targetC(2)_11}
        \frac{p\nu_M(F'_1)+\nu_M(F'_2)}{\nu_M(F'_0)+\nu_M(F'_1)+\nu_M(F'_2)} = r.
    \end{equation}
    
    \emph{Case 1:} $\tilde{\eta}$ ends with $\tilde{\eta}_{n-1} \tilde{\eta}_n=1^2$ (Table~\ref{tab:1case1}).

    \begin{table}[H]
    \caption{Transition probabilities and patterns (Case 1)}\label{tab:1case1}
    \centering
    \begin{tabular}{rcc}
        \toprule
        $Q_n$ & $\eta$ & Probability/Common factor \\
        \midrule
        $q_0$ & $\cdots\texttt{100}$ & $p(1-r)$ \\
        $q_1$ & $\cdots\texttt{110}$ & $r(1-r)$ \\
        $q_2$ & $\cdots\texttt{111}$ & $r^2$ \\
        \bottomrule
    \end{tabular}
    \end{table}
    
    \emph{Case 2:} $\tilde{\eta}$ ends with $\tilde{\eta}_{n-1} \tilde{\eta}_n=01$ and when the carrier runs from right to left, it encounters $``00"$ before $``11"$ in $\tilde{\eta}$ (Table~\ref{tab:1case2}).

    \begin{table}[htbp]
    \caption{Transition probabilities and patterns (Case 2)}\label{tab:1case2}
    \centering
    \begin{tabular}{rcc}
        \toprule
        $Q_n$ & $\eta$ & Probability/Common factor \\
        \midrule
        $q_0$ & $\texttt{001}\cdots\texttt{10}$ & $p(1-p)(1-r)$ \\
        $q_1$ & $\texttt{011}\cdots\texttt{10}$ & $(1-p)r(1-r)$ \\
        $q_2$ & $\texttt{011}\cdots\texttt{11}$ & $(1-p)r^2$ \\
        \bottomrule
    \end{tabular}
    \end{table}
    
    \emph{Case 3:} $\tilde{\eta}$ ends in $\tilde{\eta}_{n-1} \tilde{\eta}_n=01$ and when the carrier runs from right to left, it encounters $``11"$ before $``00"$ in $\tilde{\eta}$ (Table~\ref{tab:1case3}).

    \begin{table}[htbp]
    \caption{Transition probabilities and patterns (Case 3)}\label{tab:1case3}
    \centering
    \begin{tabular}{rcc}
        \toprule
        $Q_n$ & $\eta$ & Probability/Common factor \\
        \midrule
        $q_0$ & $\texttt{100}\cdots\texttt{10}$ & $p(1-r)^2$ \\
        $q_1$ & $\texttt{110}\cdots\texttt{10}$ & $r(1-r)^2$ \\
        $q_2$ & $\texttt{110}\cdots\texttt{11}$ & $r^2(1-r)$ \\
        \bottomrule
    \end{tabular}
    \end{table}
    
    In all cases, $\nu_M(F'_0) : \nu_M(F'_1) : \nu_M(F'_2) = p(1-r) : r(1-r) : r^2$. Substituting this constant ratio into~\eqref{targetC(2)_11} confirms the equality, which implies~\eqref{eq:C2Markovr}.
\end{proof}

\subsection{Invariant Measures for BBS-S(2) and BBS-V(2)}

\begin{theorem}\label{inv_meas_S2V2}
    For $p,r \in (0,1)$, $\nu_M$ is an invariant measure for the systems \textup{BBS-S(2)} and \textup{BBS-V(2)} if and only if $r= 1-p$. This condition corresponds to the case where $\nu_M$ is the Bernoulli product measure $\nu_p$.
\end{theorem}
\begin{proof}
    From Tables~\ref{tab:S2} and~\ref{tab:V2}, the automata of \textup{BBS-S(2)} and \textup{BBS-V(2)} satisfy~\ref{enum:bijective}, \ref{enum:unique_final_state}, and \ref{enum:reset_word} with $q^* = q_0$ and $N = 2$. Moreover, they are both particle-preserving with the same weight function
    \begin{equation}
    \left\{
        \begin{aligned}
            &w(0)=w(q_0)=0 \\
            &w(1) = w(q_1) = w(q_2)=1.
        \end{aligned}
    \right.
    \end{equation}
    Indeed, Tables~\ref{tab:particle_preservingS(2)} and~\ref{tab:particle_preservingV(2)} verify the particle-preserving property for BBS-S(2) and BBS-V(2). Therefore, by Theorem~\ref{thm:Ber_inv_meas}, $\nu_p$ is an invariant measure for BBS-S(2) and BBS-V(2) for $p\in(0,1)$.
    
    \begin{table}[htbp]
    \caption{Verification of particle preservation for BBS-S(2)}\label{tab:particle_preservingS(2)}
    \centering
    \begin{tabular}{c|c|c}
    \toprule
    Transition $(q,s)\mapsto(q',s')=(\varphi(q,s),\psi(q,s))$ & $w(q)+w(s)$ & $w(q')+w(s')$\\
    \midrule
    $(q_0,0)\mapsto(q_0,0)$ & $0$ & $0$ \\
    $(q_1,0)\mapsto(q_2,0)$ & $1$ & $1$ \\
    $(q_2,0)\mapsto(q_0,1)$ & $1$ & $1$ \\
    $(q_0,1)\mapsto(q_1,0)$ & $1$ & $1$ \\
    $(q_1,1)\mapsto(q_2,1)$ & $2$ & $2$ \\
    $(q_2,1)\mapsto(q_1,1)$ & $2$ & $2$ \\
    \bottomrule
    \end{tabular}
    \end{table}
    
    \begin{table}[H]
    \caption{Verification of particle preservation for BBS-V(2)}
    \label{tab:particle_preservingV(2)}
    \centering
    \begin{tabular}{c|c|c}
    \toprule
    Transition $(q,s)\mapsto(q',s')=(\varphi(q,s),\psi(q,s))$ & $w(q)+w(s)$&$w(q')+w(s')$ \\
    \midrule
    $(q_0,0)\mapsto(q_0,0)$ & $0$ & $0$ \\
    $(q_1,0)\mapsto(q_2,0)$ & $1$ & $1$ \\
    $(q_2,0)\mapsto(q_0,1)$ & $1$ & $1$ \\
    $(q_0,1)\mapsto(q_1,0)$ & $1$ & $1$ \\
    $(q_1,1)\mapsto(q_1,1)$ & $2$ & $2$ \\
    $(q_2,1)\mapsto(q_2,1)$ & $2$ & $2$ \\
    \bottomrule
    \end{tabular}
    \end{table}
    Next, we show that if the input process $\eta$ is a Markov chain with $r\neq1-p$, then the output process $\tilde{\eta}$ is not a Markov chain. To prove this, we show that the transition probability to the next state $\tilde{\eta}_{n+1}$ depends not only on the current state $\tilde{\eta}_n$ but also on the past.
    
        \emph{Case 1:} BBS-S(2) \\
        For BBS-S(2), a direct computation shows that
        \begin{align}
            \nu_M(\tilde{\eta}_{n+1}=0\mid \tilde{\eta}_{n-4}\cdots\tilde{\eta}_n=0^5)&= p \label{eq:0^6}\\
            \nu_M(\tilde{\eta}_{n+1}=0\mid \tilde{\eta}_{n-4}\cdots\tilde{\eta}_n=0^21^20 )&= \frac{pr-r+1}{r+1}.\label{eq:0^31^20}
        \end{align}
        For~\eqref{eq:0^6}, we have
        \begin{equation}
        \left\{
        \begin{aligned}
            \tilde{\psi}(0^6;q_0) &= 0^6 \\
            \tilde{\psi}(0^6;q_1) &= 0^51 \\
            \tilde{\psi}(0^6;q_2) &= 0^410 \\
            \tilde{\psi}(0^5;q_0) &= 0^5 \\
            \tilde{\psi}(0^5;q_1) &= 0^41 \\
            \tilde{\psi}(0^5;q_2) &= 0^310.
        \end{aligned}
        \right.
        \end{equation}
        Therefore, we obtain 
        \begin{align}
            \nu_M(\tilde{\eta}_{n+1}=0\mid \tilde{\eta}_{n-4}\cdots\tilde{\eta}_n=0^5)
            &= \frac{\sum_{k=0}^2\nu_M(\tilde{\eta}_{n-4}\cdots\tilde{\eta}_{n+1}=0^6,Q_{n+1} = q_k)}{\sum_{k=0}^2\nu_M(\tilde{\eta}_{n-4}\cdots\tilde{\eta}_{n}=0^5, Q_n = q_k)} \\
            &=\frac{p^5+p^4(1-p)+p^3(1-p)(1-r)}{p^4+p^3(1-p)+p^2(1-p)(1-r)}\\
            &= p.
        \end{align}
        For~\eqref{eq:0^31^20}, on the other hand, from
        \begin{equation}
            \left\{
            \begin{aligned}
                \tilde{\psi}(0^21^20^2;q_0) &= 01^20^3 \\
                \tilde{\psi}(0^21^20^2;q_1) &= 01^20^21 \\
                \tilde{\psi}(0^21^20^2;q_2) &= 01^2010 \\
                \tilde{\psi}(0^21^20;q_0) &= 01^20^2 \\
                \tilde{\psi}(0^21^20;q_1) &= 01^201 \\
                \tilde{\psi}(0^21^20;q_2) &= 01^30,
            \end{aligned}
            \right.
        \end{equation}
        we obtain
        \begin{align}
            \nu_M(\tilde{\eta}_{n+1}=0\mid \tilde{\eta}_{n-4}\cdots\tilde{\eta}_n=0^21^20 )
            &= \frac{p^2(1-p)r(1-r) + p(1-p)^2r(1-r) + (1-p)^2r(1-r)^2}{p(1-p)r(1-r) + (1-p)^2r(1-r) + (1-p)r^2(1-r)}\\
            &= \frac{pr-r+1}{r+1}.
        \end{align}
        
        \emph{Case 2:} BBS-V(2) \\
        For BBS-V(2), the following identities hold:
        \begin{align}
            \nu_M(\tilde{\eta}_{n+1}=0\mid \tilde{\eta}_{n-4}\cdots\tilde{\eta}_n=0^5)&= p\label{eq:0^5}\\
            \nu_M(\tilde{\eta}_{n+1}=0\mid \tilde{\eta}_{n-4}\cdots\tilde{\eta}_n=0^310)&= \frac{pr-r+1}{r+1}.\label{eq:0^310}
        \end{align}
        For~\eqref{eq:0^5}, we have
        \begin{equation}
            \left\{
            \begin{aligned}
                \tilde{\psi}(0^6;q_0) &= 0^6 \\
                \tilde{\psi}(0^6;q_1) &= 0^51 \\
                \tilde{\psi}(0^6;q_2) &= 0^410 \\
                \tilde{\psi}(0^5;q_0) &= 0^5 \\
                \tilde{\psi}(0^5;q_1) &= 0^41 \\
                \tilde{\psi}(0^5;q_2) &= 0^310.
            \end{aligned}
            \right.
        \end{equation}
        We obtain 
        \begin{align}
            \nu_M(\tilde{\eta}_{n+1}=0\mid \tilde{\eta}_{n-4}\cdots\tilde{\eta}_n=0^5)
            &= \frac{p^5 + p^4(1-p) + p^3(1-p)(1-r)}{p^4 + p^3(1-p) + p^2(1-p)(1-r)} \\
            &= p.
        \end{align}
        For~\eqref{eq:0^310}, we have
        \begin{equation}
        \left\{
            \begin{aligned}
                \tilde{\psi}(0^310^2;q_0) &=010^4 \\
                \tilde{\psi}(0^310^2;q_1) &=010^31 \\
                \tilde{\psi}(0^310^2;q_2) &=010^210 \\
                \tilde{\psi}(0^310;q_0) &=010^3 \\
                \tilde{\psi}(0^310;q_1) &=010^21 \\
                \tilde{\psi}(0^310;q_2) &=0^21^20.
            \end{aligned}
        \right.
        \end{equation}
        Thus, we obtain
        \begin{align}
            \nu_M(\tilde{\eta}_{n+1}=0\mid \tilde{\eta}_{n-4}\cdots\tilde{\eta}_n=0^310)
            &= \frac{p^3(1-p)(1-r) + p^2(1-p)^2(1-r) + p(1-p)^2(1-r)^2}{p^2(1-p)(1-r) + p(1-p)^2(1-r) + p(1-p)r(1-r)}\\
            &=\frac{pr-r+1}{r+1}.
        \end{align}
    Since
    \begin{equation}
        \frac{pr-r+1}{r+1} = p \Leftrightarrow r = 1-p,
    \end{equation}
    the condition $r=1-p$ is necessary in both cases for $\tilde{\eta}$ to be a Markov process.
\end{proof}

\section*{Acknowledgments}

We thank Makiko Sasada for fruitful discussions and helpful comments on the manuscript. We are also grateful to Satoshi Tsujimoto for insightful comments on integrable systems. We thank Hayate Suda and Kanta Nozawa for their valuable feedback during seminars.
T.K. acknowledges support from the Forefront Physics and Mathematics Program to Drive Transformation (FoPM), a World-leading Innovative Graduate Study (WINGS) Program, at the University of Tokyo.

\appendix

\section{Proof of Proposition~\ref{prop:Mar_dist_supp}}\label{app:Mar_dist_supp}

\begin{proof}
Since the stationary distribution of $\{\eta_n\}_{n\in \mathbb{Z}}$ is reversible, the sequences $\{\eta_n\}_{n\in\mathbb{Z}}$ and $\{\eta_{-n}\}_{n\in\mathbb{Z}}$ have the same distribution. Therefore, the proof for $\Omega_-$ follows immediately from that for $\Omega_+$. Then, it suffices to show that for any $N\in \mathbb{Z}_{\geq1}$ and $M\in \mathbb{Z}_{\geq 1}$,
\begin{equation}
\nu_M\left(
\bigcap_{m = M}^\infty A_{m, N}^c
\right)
 = 0.
\end{equation}
By stationarity, this is equivalent to 
\begin{equation}
    \nu_M\left(
    \bigcap_{m  =0}^\infty A_{m,N}^c\right) = 0.
\end{equation}
We define $\delta$ as the minimum conditional probability of $A_{m, N}$ given the state preceding the block:
\begin{equation}
    \delta = \nu_M(A_{m,N} \mid \eta_{m-1} = 0)\wedge \nu_M(A_{m,N}\mid \eta_{m-1}=1).
\end{equation}
Note that $\delta$ is a positive constant independent of $m$. For $K \in \mathbb{Z}_{\geq 1}$, we have
\begin{equation}
    \nu_M\left(A_{KN, N}^c\left| \bigcap_{k = 0}^{K-1} A_{kN, N}^c\right. \right) \leq 1-\delta .
\end{equation}
Therefore,
\begin{align}
    \nu_M\left( \bigcap_{m = 0}^\infty A_{m,N}^c \right)
    &\leq \lim_{K \to \infty} \nu_M\left(\bigcap_{k = 0}^K A_{kN, N}^c \right)\\
    &= \lim_{K \to \infty}\nu_M(A_{0,N}^c)\prod_{k=1}^K \nu_M\left(A_{kN, N}^c\left|\bigcap_{l=0}^{k-1}A_{lN, N}^c \right.\right) \\
    &\leq \lim_{K \to \infty} (1-\delta)^{K} \\
    &= 0.
\end{align}
\end{proof}

\section{Proof of Remark~\ref{rem:Omega_ext}}
\label{app:Omega_ext}

\begin{proof}
    First we show that $\Omega_{\text{ext}}\subset \Omega_{\mathscr{A}}$ and $T_\mathscr{A}(\Omega_\text{ext})\subset \Omega_\text{ext}$ for $\mathscr{A}=\mathscr{A}_{\textup{BBS-C(2)}}$. $\Omega_\text{ext}\subset \Omega_\mathscr{A}$ holds since $\mathscr{A}$ satisfies~\ref{enum:reset_word} with $N=2$. To prove the second inclusion, fix $\eta\in \Omega_\text{ext}$ and set $\tilde{\eta} = T_{\mathscr{A}}(\eta)$. We need to show that $\tilde{\eta} \in \Omega_{\text{ext}}$. Let $Q_n$ denote the state of the carrier after passing site $n$. It suffices to show that for any positive integer $M$, there exists an integer $m\geq M$ such that $\tilde{\eta}_m = \tilde{\eta}_{m + 1} = 0$. The proof for the negative direction is analogous. If $\eta$ is eventually zero to the right, then the claim is immediate because the carrier returns to the initial state $q_0$ after reading the reset word $``00"$. Hence we may assume that $\eta$ contains infinitely many $1$'s to the right. Since $\eta \in \Omega_{\text{ext}}$, there exists an integer $m_0\geq M$ such that $\eta_{m_0} = \eta_{m_0 + 1} =0$ and $\eta_{m_0 + 2} = 1$. Let $m_1$ be the smallest integer $n> m_0$ such that $\eta_n=\eta_{n + 1} = 0$. We consider the following three cases according to the carrier state $Q_{m_1-1}$ upon arriving at site $m_1$.
    \begin{itemize}
        \item $Q_{m_1-1} = q_0$:\\
        In this case, $\eta_{m_1-1}$ must be 0, which contradicts the minimality of $m_1$. Hence this case cannot occur.
        \item $Q_{m_1-1} = q_1$:\\
        In this case, we have $\tilde{\eta}_{m_1+1} = 0$ and the carrier state after passing site $m_1 + 1$ is $Q_{m_1+1} = q_0$. Therefore, it follows that $\tilde{\eta}_{m_1+2} = 0$.
        \item $Q_{m_1-1} = q_2$:\\
        Let $n_0$ be the smallest integer $n\geq m_0$ such that $Q_n = q_2$. The output at this site is $\tilde{\eta}_{n_0} = 0$, and the state of the carrier upon arriving at $n_0$ is $Q_{n_0-1} = q_1$. Noting that the carrier cannot be in state $q_1$ for two consecutive sites, we have $Q_{n_0-2}=q_0$. It follows that the carrier must have collected balls at sites $n_0-1$ and $n_0$, and thus we see that $\tilde{\eta}_{n_0-1} = \tilde{\eta}_{n_0} = 0$.
    \end{itemize}
    In each case, we have shown the existence of the pattern $``00"$ in the evolved configuration $\tilde{\eta}$ at a position greater than or equal to $M$.
        
    Next, we show that the configuration space $\Omega_{\textup{ext}}$ is not closed under the time evolution of \textup{BBS-S(2)} and \textup{BBS-V(2)}. We first note that the inclusion $\Omega_{\textup{ext}}\subset\Omega_\mathscr{A}$ still holds for $\textup{BBS-S(2)}$ and $\textup{BBS-V(2)}$ since the carrier returns to $q_0$ after reading $``00"$. However, the following examples in Figs.~\ref{fig:S(2)} and \ref{fig:V(2)} illustrate that $``00"$ sequences are not necessarily conserved under the time evolution. In each example, the configuration at time $t=0$ contains infinitely many occurrences of $``00"$. However, after one time step, the configuration at time $t=1$ contains only a finite number of $``00"$ sequences. This shows that the configuration space $\Omega_{\textup{ext}}$ is not closed under the time evolution of \textup{BBS-S(2)} and \textup{BBS-V(2)}.
    
    \begin{figure}[htbp]
        \centering
        \[
        \begin{array}{ll}
            \text{t=0:} & \texttt{001110011100111001110011100111} \\
            \text{t=1:} & \texttt{010110101101011010110101101011} \\
        \end{array}
        \]
        \caption{An example illustrating the non-conservation of $``00"$ sequences in BBS-S(2)}
        \label{fig:S(2)}
    \end{figure}
    
    \begin{figure}[htbp]
        \centering
        \[
        \begin{array}{ll}
            \text{t=0:} & \texttt{001100110011001100110011001100} \\
            \text{t=1:} & \texttt{010101010101010101010101010101} \\
        \end{array}
        \]
        \caption{An example illustrating the non-conservation of $``00"$ sequences in BBS-V(2)}
        \label{fig:V(2)}
    \end{figure}
\end{proof}


\begin{thebibliography}{13}%
\makeatletter
\providecommand \@ifxundefined [1]{%
 \@ifx{#1\undefined}
}%
\providecommand \@ifnum [1]{%
 \ifnum #1\expandafter \@firstoftwo
 \else \expandafter \@secondoftwo
 \fi
}%
\providecommand \@ifx [1]{%
 \ifx #1\expandafter \@firstoftwo
 \else \expandafter \@secondoftwo
 \fi
}%
\providecommand \natexlab [1]{#1}%
\providecommand \enquote  [1]{``#1''}%
\providecommand \bibnamefont  [1]{#1}%
\providecommand \bibfnamefont [1]{#1}%
\providecommand \citenamefont [1]{#1}%
\providecommand \href@noop [0]{\@secondoftwo}%
\providecommand \href [0]{\begingroup \@sanitize@url \@href}%
\providecommand \@href[1]{\@@startlink{#1}\@@href}%
\providecommand \@@href[1]{\endgroup#1\@@endlink}%
\providecommand \@sanitize@url [0]{\catcode `\\12\catcode `\$12\catcode `\&12\catcode `\#12\catcode `\^12\catcode `\_12\catcode `\%12\relax}%
\providecommand \@@startlink[1]{}%
\providecommand \@@endlink[0]{}%
\providecommand \url  [0]{\begingroup\@sanitize@url \@url }%
\providecommand \@url [1]{\endgroup\@href {#1}{\urlprefix }}%
\providecommand \urlprefix  [0]{URL }%
\providecommand \Eprint [0]{\href }%
\providecommand \doibase [0]{https://doi.org/}%
\providecommand \selectlanguage [0]{\@gobble}%
\providecommand \bibinfo  [0]{\@secondoftwo}%
\providecommand \bibfield  [0]{\@secondoftwo}%
\providecommand \translation [1]{[#1]}%
\providecommand \BibitemOpen [0]{}%
\providecommand \bibitemStop [0]{}%
\providecommand \bibitemNoStop [0]{.\EOS\space}%
\providecommand \EOS [0]{\spacefactor3000\relax}%
\providecommand \BibitemShut  [1]{\csname bibitem#1\endcsname}%
\let\auto@bib@innerbib\@empty
\bibitem [{\citenamefont {Novikov}\ \emph {et~al.}(1984)\citenamefont {Novikov}, \citenamefont {Manakov}, \citenamefont {Pitaevskii},\ and\ \citenamefont {Zakharov}}]{novikov1984theory}%
  \BibitemOpen
  \bibfield  {author} {\bibinfo {author} {\bibfnamefont {S.}~\bibnamefont {Novikov}}, \bibinfo {author} {\bibfnamefont {S.~V.}\ \bibnamefont {Manakov}}, \bibinfo {author} {\bibfnamefont {L.~P.}\ \bibnamefont {Pitaevskii}},\ and\ \bibinfo {author} {\bibfnamefont {V.~E.}\ \bibnamefont {Zakharov}},\ }\href@noop {} {\emph {\bibinfo {title} {Theory of solitons: the inverse scattering method}}}\ (\bibinfo  {publisher} {Springer},\ \bibinfo {address} {New York},\ \bibinfo {year} {1984})\BibitemShut {NoStop}%
\bibitem [{\citenamefont {Walters}(2000)}]{walters2000introduction}%
  \BibitemOpen
  \bibfield  {author} {\bibinfo {author} {\bibfnamefont {P.}~\bibnamefont {Walters}},\ }\href@noop {} {\emph {\bibinfo {title} {An introduction to ergodic theory}}},\ Vol.~\bibinfo {volume} {79}\ (\bibinfo  {publisher} {Springer},\ \bibinfo {address} {New York},\ \bibinfo {year} {2000})\BibitemShut {NoStop}%
\bibitem [{\citenamefont {Friedli}\ and\ \citenamefont {Velenik}(2017)}]{friedli2017statistical}%
  \BibitemOpen
  \bibfield  {author} {\bibinfo {author} {\bibfnamefont {S.}~\bibnamefont {Friedli}}\ and\ \bibinfo {author} {\bibfnamefont {Y.}~\bibnamefont {Velenik}},\ }\href {https://doi.org/10.1017/9781316882603} {\emph {\bibinfo {title} {Statistical mechanics of lattice systems: a concrete mathematical introduction}}}\ (\bibinfo  {publisher} {Cambridge University Press},\ \bibinfo {address} {Cambridge},\ \bibinfo {year} {2017})\BibitemShut {NoStop}%
\bibitem [{\citenamefont {Spohn}(2020)}]{spohn2020generalized}%
  \BibitemOpen
  \bibfield  {author} {\bibinfo {author} {\bibfnamefont {H.}~\bibnamefont {Spohn}},\ }\bibfield  {title} {\bibinfo {title} {Generalized {G}ibbs ensembles of the classical {T}oda chain},\ }\href {https://doi.org/10.1007/s10955-019-02320-5} {\bibfield  {journal} {\bibinfo  {journal} {Journal of Statistical Physics}\ }\textbf {\bibinfo {volume} {180}},\ \bibinfo {pages} {4} (\bibinfo {year} {2020})}\BibitemShut {NoStop}%
\bibitem [{\citenamefont {Kuniba}\ \emph {et~al.}(2020)\citenamefont {Kuniba}, \citenamefont {Misguich},\ and\ \citenamefont {Pasquier}}]{kuniba2020generalized}%
  \BibitemOpen
  \bibfield  {author} {\bibinfo {author} {\bibfnamefont {A.}~\bibnamefont {Kuniba}}, \bibinfo {author} {\bibfnamefont {G.}~\bibnamefont {Misguich}},\ and\ \bibinfo {author} {\bibfnamefont {V.}~\bibnamefont {Pasquier}},\ }\bibfield  {title} {\bibinfo {title} {Generalized hydrodynamics in box-ball system},\ }\href {https://doi.org/10.1088/1751-8121/abadb9} {\bibfield  {journal} {\bibinfo  {journal} {Journal of Physics A: Mathematical and Theoretical}\ }\textbf {\bibinfo {volume} {53}},\ \bibinfo {pages} {404001} (\bibinfo {year} {2020})}\BibitemShut {NoStop}%
\bibitem [{\citenamefont {Doyon}(2019)}]{doyon2019generalized}%
  \BibitemOpen
  \bibfield  {author} {\bibinfo {author} {\bibfnamefont {B.}~\bibnamefont {Doyon}},\ }\bibfield  {title} {\bibinfo {title} {Generalized hydrodynamics of the classical {T}oda system},\ }\href {https://doi.org/10.1063/1.5096892} {\bibfield  {journal} {\bibinfo  {journal} {Journal of Mathematical Physics}\ }\textbf {\bibinfo {volume} {60}},\ \bibinfo {pages} {073302} (\bibinfo {year} {2019})}\BibitemShut {NoStop}%
\bibitem [{\citenamefont {Takahashi}\ and\ \citenamefont {Satsuma}(1990)}]{19903514}%
  \BibitemOpen
  \bibfield  {author} {\bibinfo {author} {\bibfnamefont {D.}~\bibnamefont {Takahashi}}\ and\ \bibinfo {author} {\bibfnamefont {J.}~\bibnamefont {Satsuma}},\ }\bibfield  {title} {\bibinfo {title} {A soliton cellular automaton},\ }\href {https://doi.org/10.1143/JPSJ.59.3514} {\bibfield  {journal} {\bibinfo  {journal} {Journal of the Physical Society of Japan}\ }\textbf {\bibinfo {volume} {59}},\ \bibinfo {pages} {3514} (\bibinfo {year} {1990})}\BibitemShut {NoStop}%
\bibitem [{\citenamefont {Takahashi}\ and\ \citenamefont {Matsukidaira}(1997)}]{DaisukeTakahashi_1997}%
  \BibitemOpen
  \bibfield  {author} {\bibinfo {author} {\bibfnamefont {D.}~\bibnamefont {Takahashi}}\ and\ \bibinfo {author} {\bibfnamefont {J.}~\bibnamefont {Matsukidaira}},\ }\bibfield  {title} {\bibinfo {title} {Box and ball system with a carrier and ultradiscrete modified {K}d{V} equation},\ }\href {https://doi.org/10.1088/0305-4470/30/21/005} {\bibfield  {journal} {\bibinfo  {journal} {Journal of Physics A: Mathematical and General}\ }\textbf {\bibinfo {volume} {30}},\ \bibinfo {pages} {L733} (\bibinfo {year} {1997})}\BibitemShut {NoStop}%
\bibitem [{\citenamefont {Croydon}\ \emph {et~al.}(2023)\citenamefont {Croydon}, \citenamefont {Kato}, \citenamefont {Sasada},\ and\ \citenamefont {Tsujimoto}}]{croydon2023dynamics}%
  \BibitemOpen
  \bibfield  {author} {\bibinfo {author} {\bibfnamefont {D.}~\bibnamefont {Croydon}}, \bibinfo {author} {\bibfnamefont {T.}~\bibnamefont {Kato}}, \bibinfo {author} {\bibfnamefont {M.}~\bibnamefont {Sasada}},\ and\ \bibinfo {author} {\bibfnamefont {S.}~\bibnamefont {Tsujimoto}},\ }\href {https://doi.org/10.1090/memo/1398} {\emph {\bibinfo {title} {Dynamics of the box-ball system with random initial conditions via Pitman’s transformation}}},\ Vol.\ \bibinfo {volume} {283}\ (\bibinfo  {publisher} {American Mathematical Society},\ \bibinfo {address} {Providence, RI},\ \bibinfo {year} {2023})\BibitemShut {NoStop}%
\bibitem [{\citenamefont {Ferrari}\ \emph {et~al.}(2021)\citenamefont {Ferrari}, \citenamefont {Nguyen}, \citenamefont {Rolla},\ and\ \citenamefont {Wang}}]{ferrari2021soliton}%
  \BibitemOpen
  \bibfield  {author} {\bibinfo {author} {\bibfnamefont {P.~A.}\ \bibnamefont {Ferrari}}, \bibinfo {author} {\bibfnamefont {C.}~\bibnamefont {Nguyen}}, \bibinfo {author} {\bibfnamefont {L.~T.}\ \bibnamefont {Rolla}},\ and\ \bibinfo {author} {\bibfnamefont {M.}~\bibnamefont {Wang}},\ }\bibfield  {title} {\bibinfo {title} {Soliton decomposition of the box-ball system},\ }\href {https://doi.org/10.1017/fms.2021.49} {\bibfield  {journal} {\bibinfo  {journal} {Forum of Mathematics, Sigma}\ }\textbf {\bibinfo {volume} {9}},\ \bibinfo {pages} {e60} (\bibinfo {year} {2021})}\BibitemShut {NoStop}%
\bibitem [{\citenamefont {Croydon}\ and\ \citenamefont {Sasada}(2021)}]{Croydon2021}%
  \BibitemOpen
  \bibfield  {author} {\bibinfo {author} {\bibfnamefont {D.~A.}\ \bibnamefont {Croydon}}\ and\ \bibinfo {author} {\bibfnamefont {M.}~\bibnamefont {Sasada}},\ }\bibfield  {title} {\bibinfo {title} {Generalized hydrodynamic limit for the box--ball system},\ }\href {https://doi.org/10.1007/s00220-020-03914-x} {\bibfield  {journal} {\bibinfo  {journal} {Communications in Mathematical Physics}\ }\textbf {\bibinfo {volume} {383}},\ \bibinfo {pages} {427} (\bibinfo {year} {2021})}\BibitemShut {NoStop}%
\bibitem [{\citenamefont {Mealy}(1955)}]{Mealy1955}%
  \BibitemOpen
  \bibfield  {author} {\bibinfo {author} {\bibfnamefont {G.~H.}\ \bibnamefont {Mealy}},\ }\bibfield  {title} {\bibinfo {title} {A method for synthesizing sequential circuits},\ }\href {https://doi.org/10.1002/j.1538-7305.1955.tb03788.x} {\bibfield  {journal} {\bibinfo  {journal} {The Bell System Technical Journal}\ }\textbf {\bibinfo {volume} {34}},\ \bibinfo {pages} {1045} (\bibinfo {year} {1955})}\BibitemShut {NoStop}%
\bibitem [{\citenamefont {Maeno}\ \emph {et~al.}(2025)\citenamefont {Maeno}, \citenamefont {Tsujimoto},\ and\ \citenamefont {Yura}}]{maeno2025solitons}%
  \BibitemOpen
  \bibfield  {author} {\bibinfo {author} {\bibfnamefont {A.}~\bibnamefont {Maeno}}, \bibinfo {author} {\bibfnamefont {S.}~\bibnamefont {Tsujimoto}},\ and\ \bibinfo {author} {\bibfnamefont {F.}~\bibnamefont {Yura}},\ }\bibfield  {title} {\bibinfo {title} {Solitons in 3-state {M}ealy automata},\ }\href {https://doi.org/10.1016/j.physd.2025.134857} {\bibfield  {journal} {\bibinfo  {journal} {Physica D: Nonlinear Phenomena}\ }\textbf {\bibinfo {volume} {481}},\ \bibinfo {pages} {134857} (\bibinfo {year} {2025})}\BibitemShut {NoStop}%
\end{thebibliography}
\end{document}